\documentclass[12pt]{article}
\usepackage{amsmath,amssymb,mathrsfs,framed}
\usepackage{graphicx}

\usepackage[colorlinks]{hyperref}
\usepackage{color}
\usepackage[all]{xy}
\usepackage{amsthm}
\usepackage[margin=2cm]{geometry}

\newtheorem{theorem}{Theorem}[section]

\newtheorem{lemma}[theorem]{Lemma}
\newtheorem{remark}[theorem]{Remark}
\newtheorem{proposition}[theorem]{Proposition}
\newtheorem{corollary}[theorem]{Corollary}

\newcommand{\pp}[2]{\frac{\partial #1}{\partial #2}}

\newcommand{\dede}[2]{\frac{\delta #1}{\delta #2}}

\newcommand{\comment}[1]{\vspace{5mm}\par
\framebox{\begin{minipage}[c]{.95 \textwidth} \rm #1
\end{minipage}}\vspace{5 mm}\par}

\newcommand{\rem}[1]{}
\newcommand{\Diff}{\mathrm{Diff}}
\newcommand{\de}{\mathrm{d}}

\newcommand{\z}{{\mathbf{\hat{z}}}}
\newcommand{\bx}{{\mathbf{x}}}
\newcommand{\bq}{{\mathbf{q}}}
\newcommand{\bv}{{\mathbf{v}}}

\newcommand{\bn}{{\mathbf{n}}}

\newcommand{\bE}{{\mathbf{E}}}

\newcommand{\R}{{\mathcal{R}}}

\newcommand{\bz}{{\boldsymbol{z}}}

\newcommand{\bfi}{\bfseries\itshape}
\newcommand{\bchi}{\boldsymbol{\chi}}

\newcommand{\bomega}{\boldsymbol{\omega}}

\newcommand{\bnu}{\boldsymbol{\nu}}
\newcommand{\bsigma}{\boldsymbol{\sigma}}
\newcommand{\bxi}{\boldsymbol{\xi}}

\begin{document}

\title{\vspace{-1cm}Collisionless kinetic theory of rolling molecules}
\author{
Darryl D. Holm$^{1}$, Vakhtang Putkaradze$^{2}$ and Cesare Tronci$^{3}\!$
\\
\footnotesize \!\!\!$^1$ \it Department of Mathematics, Imperial College London, London SW7 2AZ, UK\\
\footnotesize \!\!\!$^2$ \it Department of Mathematical and Statistical Sciences, University of Alberta, \\ 
\footnotesize \it Edmonton, AB T6G 2G1 Canada\\
\footnotesize \!\!\!$^3$\,\it Department of Mathematics, University of Surrey, Guildford GU2 7XH, UK
}
\date{}
\maketitle
\vspace{-.8cm}
\begin{abstract} 
We derive a collisionless kinetic theory for an ensemble of molecules undergoing nonholonomic rolling dynamics.  We demonstrate that the existence of nonholonomic constraints leads to problems in generalizing the standard methods of statistical physics. In particular, we show that even though the energy of the system is conserved, and the system is closed in the thermodynamic sense, some fundamental  features of statistical physics such as invariant measure do not hold for such nonholonomic systems.
Nevertheless, we are able to construct a consistent kinetic theory using Hamilton's variational principle in Lagrangian variables, by regarding the kinetic solution as being concentrated on the constraint distribution. A cold fluid closure for the kinetic system is also presented, along with a particular class of exact solutions of the kinetic equations. 
\end{abstract}

{\footnotesize
\tableofcontents
}

\section{Introduction} 
Constrained dynamical systems 
play an important role in modern mechanics and statistical physics. The constraints applied to the system can be separated into two broad categories -- holonomic and nonholonomic. Holonomic constraints restrict the particle motion to lie a certain surface in the configuration space. 
Nonholonomic constraints are then defined as any constraint that cannot be reduced to  motion on a particular surface in the configuration space.  
There are some classical examples of nonholonomic systems with the constraints  that are \emph{linear} in velocities. These systems usually (but not necessarily) come from perfect friction limitation, such as rolling particles \cite{Ch1903} (Chaplygin's ball) or, more broadly, a particular connection between several components of velocities, such as Chaplygin's sleigh \cite{Ch1948}. 
\rem{ 
\noindent{\color{magenta}\comment{CT: first reference to Chaplygin has been changed according to second referee. I kept the name Chaplygin ball, rather than Chaplygin's top (suggested by 2nd referee). Do you agree with my choice? If so, delete this box.}\noindent}
} 
One may also see \cite{Bl1999,BoMa2009,HoGa2009} for recent discussions of these type of problems. The Lagrange-D'Alembert principle \cite{Bl2010} is usually used to treat the dynamics of such systems.  We refer the reader to the book of Bloch \cite{Bl2010} for a discussion of nonholonomic dynamics, as well as a review of recent literature and methods. If constraints are not linear in velocities, such as isokinetic models, typically, Gauss' minimal force principle is used to derive the equations of motion. 

In keeping with the spirit of regular statistical mechanics, one would like to develop a kinetic theory for large number of coupled nonholonomic particles, akin to the Vlasov or Boltzmann theory of interacting gas particles.  However, in general, the presence of nonholonomic constraints destroys the Hamiltonian structure for the dynamics of individual particles. There are exceptions when the Hamiltonian structure of the dynamics can be restored, but these are special cases \cite{BoMa2001,BlRo2008}. Without an underlying Hamiltonian structure, the development of the kinetic theory for nonholonomically constrained systems is a formidable challenge. 

There certainly have been substantial developments in the study of the statistical physics of systems with nonholonomic constraints --the \emph{isokinetic models} -- which enforce the  constant  temperature condition for molecular particle simulations. 
 The isokinetic restriction, which is quadratic in the particle velocities, cannot be solved by either the classical Lagrange-D'Alembert method or its generalizations such as the Hamilton-Pontryagin method \cite{Ho2008}. Instead, the methods of minimal constraints due to Gauss has been used   to describe the system; see \cite{BoLeLa1999,KuJo2000,CoEzWi2010} for  for some recent progress and a review of the literature. A short discussion of this progress is warranted here. 

If, in the absence of constraints, the microscopic   particle motion is described   in phase space  by the equation $ \dot z=X(z)$, then the corresponding transport equation for the distribution function $ f(z,t)$ is  taken to be of the form 
\begin{equation}
\frac{\partial f}{\partial t}+{ \rm div}_z X(z) f =0 
\label{kineqsimple}
\end{equation}
If it is assumed that the vector field $ X(z)$ is Hamiltonian with $z=(q,p)$ and 
$ X=(H_p,-H_q)$, in which case it  has zero divergence. The standard methods of statistical mechanics then apply, in particular, the conservation law of entropy $S=\int f \log f$ holds. 

In case when vector field is not Hamiltonian, the situation is more complex. As was  first realized in \cite{Ra1986} (without proof), set in differential-geometric context in  \cite{TuMuKl1997,TuMuMa1999} and further extended in \cite{Tu-etal-2001,Ra2002,Ez2004,SeGi2007}, the non-Hamiltonian vector fields lead, generally, to curved geometry of phase space. It was shown that a more advantageous, and geometrically correct, version of transport equation is obtained by introducing the metric $ \sqrt{g}= 
\exp \int \mbox{div}_{\!z\,} X(z)\, \de t$ into (\ref{kineqsimple}) as follows:
\begin{equation} 
\frac{\partial }{\partial t} \left( \sqrt{g} f  \big) +{ \rm div}_z \big( X(z) \sqrt{g} f\right)  =0 
\label{kineqmetric} 
\end{equation} 

One can then prove that the generalized entropy $\tilde{S}=\int\! \sqrt{g}\, f \log f $ is conserved. These general considerations were then successfully applied to the systems with constraints that break the Hamiltonian structure,  in particular, to the non-holonomic isokinetic constrain enforcing constant temperature: 
\begin{equation} 
\label{isokinetic}
\sum_i \frac{1}{2} m_i \dot \bx_i^2=K_0=\mbox{const} \, . 
\end{equation} 
While this theory is general and is applicable to both linear and nonlinear constraints, a successful application of this theory hinges on the computation of the volume element $\sqrt{g(z)}$.  Unfortunately, as we show below in Remark~\ref{rem:entropy}, this approach is not applicable for the case of interacting nonholonomic particles considered here. First, the explicit computation of the volume element is not possible, and second, and more important, there are persistent fundamental difficulties with proper definition of the divergences of vector field $X$ for our case.   Thus, unfortunately, we were not able to define a conservative entropy-like quantity, because the usual definition of entropy produces a divergent integral.  
This is perhaps because every particle in the ensemble is nonholonomically constrained to roll individually rather than having a single constraint for the whole system, such as 
\eqref{isokinetic}.

In this paper, we develop a nonholonomic kinetic theory for the particular case of interacting rolling particles whose centers of mass are offset from their geometrical centers and whose inertia tensors are not proportional to unity.  Only rolling nonholonomic constraint is applied to the molecules, similar to  recent work \cite{KiPu2010} which treated this system as a model for investigating the properties of molecular monolayers. In that work, the particles had the same mechanical properties as the water molecules and thus their inertia tensors did not satisfy the symmetry requirements for the Chaplygin's ball. The theory we have developed here may also be viewed as an augmentation of the recently developed \emph{stochastic nonhonomic dynamics} \cite{Hoch2010}. 
The main contributions of this paper to the literature can be divided into three main topics. 

First main theme is the \emph{derivation of the equations of motion}, and Sections~\ref{sec:derivation},\ref{sec:E-L-rolling} and \ref{sec:E-P-rolling} are devoted to that topic. We will emphasize that in principle, we could have derived the equations of motion and corresponding kinetic equation using the Gauss' method of minimal constraint.  In that case, we would need to utilize recently developed geometric extension of this theory \cite{Le1996}, valid for constraints in arbitrary spaces, and not only in $\mathbb{R}^n$. However, we believe that such a derivation, although useful, would be exceedingly cumbersome; indeed, the Gauss' method of minimal constraint force is usually used only when a small number of constraints is present in the system. In our case, it is important to emphasize that the rolling constraints are applied to \emph{every particle in the ensemble}, so we have instead used the corresponding geometric generalization of previous ideas developed by \cite{Low1958} in the context of plasma physics with a Lagrange-d'Alembert's principle.  We do not know of another work deriving these equations of motion when the constraint is applied to every particle, rather than system in general, such as (\ref{isokinetic}). While it is possible that similar equations can be derived by Gauss' principle, we believe that the utilization of geometric methods developed here leads to important consequences, which are hard to obtain using traditional approaches.

Second, and the most important contribution of the paper lies in the proof that the kinetic theory for our nonholonomic theory can be taken initially to be concentrated on the constraint set in phase space, and it will preserve this structure for all times. This observation will allow us to derive the nonholonomic kinetic equation (\ref{EPVlasov-final2}) which we believe is the main result of this paper. 

Finally, the remainder of the paper is devoted to the analysis of the derived nonholonomic kinetic theory. More precisely, in Sections~\ref{sec:E-P-rolling} and \ref{sec:cold-fluid} we find an explicit solution of the full kinetic equation, which is inspired by Poiseuille-type flows. We also show that while the momentum is not conserved, so the traditional fluid approach  based on the density, momentum and energy conservation laws  is possible in our system, we can still derive exact conservation laws that follow from the kinetic equations. We conclude by  deriving a hydrodynamic model from a cold plasma closure of the moment hierarchy.

The theory we have developed will be directly applicable to all systems having constraints that are linear in velocities. This restriction follows from the limitations of the Lagrange-d'Alembert's principle. The question immediately arises whether this theory can be generalized to include more general constraints, that are nonlinear in velocities and are applied to every particle in the ensemble. As far as we are aware, no such theory has been developed, although we believe that  \emph{the derivation of equation} is possible based on the Gauss' principle. The crucial question for our paper lies in the possibility of concentrating the density in the phase space on the constraint set, or \emph{distribution} 
as it is commonly addressed in the nonholonomic mechanics. 
\footnote{We hope no confusion arises between the unfortunate collision of terms "distribution" as defining the nonholonomic constrained set, and the "distribution" as defining solutions in generalized function space, \emph{e.g.}, $\delta$-functions. These terms are both well-established in the appropriate literatures, and we will need to use both meanings of the word in this paper. We hope that the particular meaning of this word is clear from the context.  When a clear distinction is necessary, the kinetic probability distribution  will be called ``probability density''.
}

{A small digression into the geometry of nonholonomic constrains is warranted here. The main difficulty posed by single-particle nonholonomic dynamics is that it lies on a \emph{distribution}. That is, single-particle nonholonomic dynamics takes place on only a certain subset of the position-velocity phase space \cite{Bl2010}. In contrast, Vlasov kinetic theory takes place in phase space with independent coordinates $(x,v)$, in which one should not confuse $\dot{x}$ with $v$. The constraint for the particles may be nonholonomic, but this does not imply a relation between $\dot{x}$ and $v$.
Instead, the nonholonomic constraint is imposed in phase space by assuming the probability density is defined on the whole phase space, although it is \emph{supported} on the distribution on which the nonholonomic relation holds when we set $\dot{x}=v$.

In particular, we shall prove that the probability density that is initially concentrated on the distribution, will remain concentrated on the distribution for all times. This is the crucial point of this paper.   The  work here shows that this claim holds for arbitrary constraints that are linear in velocities. We also believe that it also holds for more general constraints, as was  demonstrated recently in the context of kinetic theory for the Vicsek model in mathematical biology \cite{BoCaCa2011,Bo-etal-2011,CaCaRo2011,BoCa-2012}.  As was shown there, the \emph{nonlinear} nonholonomic constraints of the type $|\bv_i|^2=1$ for every $i$, playing important role in mathematical biology considerations, preserve the concentration of distribution on the constraint.   The open question that remains is the following. What is the most general form of nonlinear constraint that has that property? We do not know the answer to this question, which is  very interesting, but is beyond the scope of this paper.    

The system we study here illustrates the challenges that may lie ahead in developing nonholonomic statistical mechanics. For example, the system of interacting particles is closed in the thermodynamic sense, as there is no energy exchange with the substrate. However, the system is \emph{not isolated}, in the sense that its momentum is not conserved. In fact, neither the momentum of each particle, nor the total momentum of the system is conserved. In some 
particular cases, there are additional conservation laws for individual particles (Routh and Jellett integrals) and these are useful, as we will discuss below. As has been shown in numerical simulations \cite{KiPu2010}, because of the coupling between the rotational and translational degrees of freedom, nonholonomic dynamics breaks both the ergodicity hypothesis and equipartition of energy  between different degrees of freedom, and the definition and treatment of these concepts on the constraint distribution become difficult.Correspondingly, familiar concepts such as thermodynamic temperature cannot be defined  easily.  In addition, because of the lack of momentum conservation, care must be taken in deriving  fluid-like continuum mechanics for such systems. 

\section{ Rolling molecules and collisionless kinetic theories} 
\label{sec:derivation} 
\subsection{Euler-Poincar\'e dynamics of a single rolling molecule}
\label{indiv-ball-sec}
We shall start with a brief derivation of the equations for an individual Chaplygin ball system. This is a well-known problem, but recalling the derivation here will allow us to introduce some useful notation. This derivation relies on the symmetry reduction principle for nonholonomic systems first developed in \cite{BlKrMaMu1996}. For more details on this derivation and some recent results, see \cite{Bl2010,Ho2008,Schneider,Du2004}. 
Consider an unbalanced rolling ball whose center of mass is positioned at $l \bchi \in \mathbb{R}^3$ away from the geometric center in the (body) coordinate frame in the ball, where $l$ is a  length, and $\bchi\in S^2$ is a unit vector in the body frame. As the ball moves, it undergoes a rotation $\R \in SO(3)$, and a translation $\bx \in \mathbb{R}^3$. In particular, its geometric center is translated by $\bx$, and its center of mass is located by the vector $l \R \bchi$ pointing from its geometric center to its center of mass in the (spatial) stationary coordinate frame. Let us also consider an external field $\bE$ acting on the center of mass of the ball. This could be an external electric field, gravity or another external potential force acting on the particles whose potential energy is $\mathbf{E}\cdot\R\boldsymbol\chi $. 
One writes the Lagrangian of an individual ball as the difference between its kinetic and potential energy:
\begin{equation}
L(\R,\dot\R,\bx,\dot\bx)=\frac12m |\dot\bx|^2- {\sf q} \mathbf{E}\cdot\R\boldsymbol\chi+\frac12\int_\mathcal{B}\mathcal{D}(\boldsymbol{A})|\dot\R \boldsymbol{A}|^2\,\de \boldsymbol{A} \, ,
\label{lagr0}
\end{equation} 
where $\boldsymbol{A}$ is the position vector  of a point in the molecule, which in turn possesses  mass density $\mathcal{D}$ and  volume $\mathcal{B}$. Here, 
$\mathbf{E}\cdot\R\boldsymbol\chi$ is the potential energy coming from from the motion of the centre of mass at $\R \boldsymbol \chi$ in the external field $\mathbf{E}$. The position of the "charge" ${\sf q}$ is assumed to coincide with the centre of mass, which is for example true for gravity, where ${\sf q}=m$ and $\mathbf{E}=\mathbf{g}$, the acceleration of gravity. For simplicity, without the loss of generality, we set $m=1$ by choosing appropriate time and space scale, and set ${\sf q}=1$ by appropriately defining the units of field $\mathbf{E}$.  
 Here we shall follow the Euler-Poincar\'e symmetry-reduction principle \cite{HoMaRa1998}. For this purpose, we transform the Lagrangian into the following reduced spatial variables:  
\begin{eqnarray*}
&&\hbox{Director,}\quad \mathbf{n}=\R\boldsymbol\chi\in S^2;
\\
&&\hbox{Angular velocity in the spatial frame,}\quad 
{\color{black}\widehat\nu}=\dot \R \R^T\in \mathfrak{so}(3)
; \hbox{ and }
\\
&&\hbox{Tensor of inertia in the spatial frame,}\quad  j=\R J\R^T\in sym(3\times3)
\,,
\end{eqnarray*}
where  $\mathfrak{so}(3)$ denotes the  Lie algebra of antisymmetric matrices {and $J$ is the symmetric inertia tensor $J=-\int_{\mathcal{B}}\mathcal{D}(\boldsymbol{A})\!\left(\boldsymbol{A}\boldsymbol{A}^T-\left|\boldsymbol{A}\right|^2{\bf I}\right)\de\boldsymbol{A}$.} 
The last of these spatial variables $(j)$ is known as the `microinertia tensor' in the theory of micropolar media \cite{Pa2005,Er2001}.
In fact, micropolar media provide an interesting analogy for the systems considered in this paper and many of the methods used here transfer easily to the micropolar setting. One difference, however,  is that the quantity $\bn=\R\boldsymbol\chi$ is not strictly a director, since $\bn\neq-\bn$. As we shall see, {\color{black}this parity invariance under $\mathbb{Z}_2$} is broken for rolling molecules by both  the potential $\bE\cdot\bn$ and the nonholonomic constraint. 
\begin{remark}[The hat map]\label{rem-hatmap}
In defining the angular velocity variable $\bnu$, {\color{black}we use} the following \emph{hat-map} correspondence between an antisymmetric matrix  {\color{black}$\widehat{\nu}\in\mathfrak{so}(3)$} and a vector $\bnu\in\mathbb{R}^3$ with components $\nu_i$, $i=1,2,3$,
  \begin{equation} 
  \widehat\nu_{jk} = -\,\epsilon_{jkl}\nu_l
  \quad\hbox{and}\quad
  \nu_i=-\,\frac12\,\epsilon_{ijk} \widehat\nu_{jk} \, , 
  \quad i,j,k,l=1,2,3,
  \label{hatmap}
  \end{equation}
in which $\epsilon_{ijk}$ is the completely antisymmetric tensor with $\epsilon_{123}=1$. 
\end{remark}
The \emph{symmetry-reduced} Lagrangian corresponding to (\ref{lagr0}) for an individual ball may now be expressed in spatial coordinates as 
\begin{equation}
l(\bnu,j,\bn,\bx,\dot\bx)=\frac12|\dot\bx|^2+\frac12j\boldsymbol\nu\cdot\boldsymbol\nu- \mathbf{E}\cdot\mathbf{n} \, , 
\label{redlagr}
\end{equation}
in which the time derivative is denoted with an over-dot, e.g. ${d\bx}/{dt}=:\dot{\bx}$.
The rolling constraint for an individual ball reads
\begin{equation}
\dot\bx=\dot\R(  r \R^T\z+\boldsymbol\chi)
=\boldsymbol\nu\times( r \z+\mathbf{n})
:=\boldsymbol\nu\times\bsigma(\mathbf{n})
\,,\label{rollingconstr}
\end{equation}
where $\z$ is the constant spatial unit vector pointing perpendicular to the substrate and $r$ is the radius of the ball. In (\ref{rollingconstr}), we have also introduced the notation 
\begin{equation} 
\bsigma(\mathbf{n}):=
 r \z+\mathbf{n}
 \,,
\label{sigmadef}
\end{equation}
in which the spatial vector $\bsigma(\mathbf{n})$ points from the contact point of the ball to the position of its center of mass at a given time. 
\begin{remark}[Dimensionality of the rolling constraint]
 The rolling constraint (\ref{rollingconstr}), technically speaking, defines the relationship between three-dimensional vectors. However, this contraint only contains meaningful information about the motion of either the contact point or geometric center of the ball, which are related by being offset by given constant vector. Thus, we shall understand this constraint as a relationship in $TSO(3) \times T\mathbb{R}^2$. Similar considerations will apply later to the constraint applied to the motion of an assembly of particles. 
\end{remark}
\color{black} 
\begin{proposition}
The constrained Euler-Poincar\'e variational principle, i.e., Hamilton's principle for the symmetry-reduced Lagrangian,
\[
\delta\int^{t_2}_{t_1}l(\bnu,j,\bn,\bx,\dot\bx)\,\de t=0
\]
yields the following nonholonomic equations
\begin{align}
&\frac{d}{dt}\frac{\delta l }{\delta \boldsymbol{\nu}}
+ \frac{\delta l}{\delta \boldsymbol{\nu}}\times\boldsymbol{\nu}
+\!\overrightarrow{\,\left[j,\frac{\delta l}{\delta j}\right]\,}
+ \frac{\delta l }{\delta \mathbf{n} }\times \mathbf{n}
=-\left(\frac{\de}{\de t}\frac{\delta l}{\delta \dot\bx}-\frac{\delta l}{\delta \bx}\right)
\times\bsigma(\bn)  
\,, 
\label{nueq}
\\ 
&\frac{\de j}{\de t}+\big[j,\widehat\nu\big]=0
\,, \label{jeq}
\\ 
& \frac{\de\bn}{\de t} +\bn\times\bnu=0  \, . \label{neq}
\end{align}
\end{proposition}
\begin{remark}
The notation $\overrightarrow{A}_i=\epsilon_{ijk}\,A_{jk}$ for an arbitrary antisymmetric matrix $A^T=-A$ has been introduced in (\ref{nueq}). Likewise, because of the hat map in (\ref{hatmap}) we have $\bn\times\bnu=-{\widehat\nu}\,\bn$.
\end{remark}
\begin{proof}
The proof of the proposition is standard for this kind of problem, see e.g., \cite{Ho2008}, and it uses the relation $\dot\bsigma=\dot\bn$ obtained {\color{black}from the time derivative of (\ref{sigmadef})}. 
\end{proof}

\begin{corollary}
For the reduced Lagrangian (\ref{redlagr}), the dynamic equation (\ref{nueq}) may be written equivalently as
\begin{equation}
\frac{\de}{\de t}\big(j\bnu\big)+j\bnu\times\bnu+\frac12\overrightarrow{\,\left[j,\bnu\bnu^T\right]\,}-\bE\times\bn=-(\dot\bnu\times\bsigma)\times\bsigma-(\bnu\times\dot\bn)\times\bsigma\,.
\label{chaplygin-ball-eqs}
\end{equation}
Then, upon using the properties of the hat map, one may rewrite this equation as
\begin{equation}
j\dot\bnu+\bsigma\times(\bsigma\times\dot\bnu)
=:(j+\widehat\bsigma\widehat\bsigma){\dot\bnu}
=\bE \times\bn+ j\bnu\times\bnu+\bsigma\times(\bnu\times(\bnu\times\bn) )\, . 
\label{nueq2}
\end{equation}
\end{corollary}
Equation (\ref{nueq2}) differs in form from the standard equations for the rolling ball \cite{Bl2010,Ho2008}.  This is because the standard Chaplygin ball equations are written in the body frame, whereas equations (\ref{nueq2}) have been written in the spatial (fixed) frame, in terms of the time-changing moment of inertia governed by (\ref{jeq}).  Although they have been written in the spatial frame instead of the body frame, equations (\ref{nueq2}) are mathematically equivalent to the standard equations for the Chaplygin ball. 
While previous authors have considered the Chaplygin ball equations in the {\color{black}body} frame, the {\color{black}spatial} frame will be preferable for our applications later with ensembles of rolling balls, despite the necessity of considering the inertia tensor as a time-dependent variable.  

As mentioned earlier, additional conservation laws due classically to Routh and Jellett exist for an individual rolling ball with symmetry. See e.g. \cite{Ho2008} for references and discussions of these classical integrals of motion. The Jellet conservation law, in particular, will play an important role here in our considerations of  multi-particle dynamics. 
The Jellet integral
\begin{equation} 
q_j=j \bnu\cdot \bsigma(\bn)
\label{microJellet1} 
\end{equation} 
 is conserved for an individual Chaplygin ball, provided that two of the ball's inertia tensor's eigenvalues in the body (let us call them $I_1$ and $I_2$) are equal in value $I_1=I_2$, and the axis of the third inertial eigenvector is collinear with $\bchi$, connecting the center of mass and geometric center. 
 In that particular case, there are two more integrals of motion: the energy of an individual ball and the Chaplygin (Routh) integral. As was shown by Chaplygin \cite{Ch1903}, the presence of these three integrals makes the rolling ball equations completely integrable. When there are several rolling particles of Chaplygin type interacting through a central potential directed at their centers of mass (such as the Lennard-Jones potential), the Chaplygin integral is not conserved for either an individual ball or the whole system. As one might expect, the total energy of the entire system is conserved. However, more remarkably,  the Jellet integral is still preserved \emph{for every individual ball}, in spite of their interactions \cite{KiPu2010}. 

The next section introduces the main problems that are encountered when trying to build a kinetic theory of rolling molecules. As we shall see, these problems  emerge fundamentally from the absence of an invariant measure in phase space \cite{ZeBl2003}, arising because the single-particle dynamics is \emph{not} Hamiltonian, except in very special cases.\color{black}

\subsection{Preliminary considerations in kinetic theory}

In order to describe multiparticle dynamics on phase space (typically $T\Bbb{R}^3\simeq\Bbb{R}^6$), kinetic equations govern the dynamics of a probability density $\color{black}f(\boldsymbol{z},t)$ on phase space. When collisions are neglected, kinetic equations are transport equations of the type
\begin{equation}\label{GenKinEq}
\frac{\partial f}{\partial t}+\operatorname{div}_{\!\boldsymbol{z}\!}\left(f\mathbf{X}\right)=0
\,,
\end{equation}
where the vector field $\mathbf{X}$ usually contains nonlocal terms deriving from the collective interactions among particles. In the Lagrangian picture, this means that a kinetic equation is given in characteristic form as
\begin{equation}\label{GenKinEq-Lagrangian}
\frac{\de}{\de t}\left(f_t(\boldsymbol\psi(t))\,\de^k\boldsymbol\psi(t)\right)
=0
\quad\text{along}\quad
\dot{\boldsymbol\psi}(t)
=\mathbf{X}(\boldsymbol\psi(t))
\end{equation}
where $\boldsymbol\psi(t)$ is a characteristic curve in phase space. Typically, $\boldsymbol\psi(t)=(\mathbf{x}(t),\mathbf{v}(t))\in \Bbb{R}^6$ in ordinary position-velocity coordinates, so that $\mathbf{X}(\boldsymbol\psi)\in\mathbb{R}^6$ and $k=6$. Alternatively, given the phase-space curve ${\boldsymbol\psi}(\boldsymbol\psi_0,t)$ (with initial label $\boldsymbol\psi_0$) and the initial probability density $f_0(\boldsymbol\psi_0)\,\de^k\boldsymbol\psi_0$, the forward probability density function is given by the Lagrange-to-Euler map
\begin{equation}\label{L-to-E--map}
f(\boldsymbol{z},t)=\int \!f_0(\boldsymbol\psi_0)\,\delta(\boldsymbol{z}-\boldsymbol\psi(\boldsymbol\psi_0,t))\,\de^k\boldsymbol\psi_0
\,.
\end{equation}
This expression recovers the well known Klimontovich particle solution for an initial point particle $f_0(\boldsymbol\psi_0)=\delta(\boldsymbol\psi_0-\boldsymbol{z}_0)$ at the Eulerian point $\boldsymbol{z}_0$. Here, the notation $f_t(\boldsymbol\psi(t))$ {\color{black} with time dependence in $f_t$ indicated by the subscript $t$} is used in the Lagrangian representation, while $f(\boldsymbol{z},t)$ represents the corresponding Eulerian quantity.  In mathematical terms, one says that $f_t$ is obtained as the \emph{push-forward} of $f_0$ by $\boldsymbol\psi(t)$ and one writes, e.g., 
\begin{equation}\label{Push-forward}
f_t=\boldsymbol\psi_*(t)f_0
\,.
\end{equation}

In collisionless systems (i.e. systems of the type \eqref{GenKinEq}), the motion is reversible. When this motion is also \emph{Hamiltonian}, the flow preserves the Liouville volume on phase space $\de^k\boldsymbol\psi_0$, so that $\de^k\boldsymbol\psi(t)=\de^k\boldsymbol\psi_0$ and hence the relation \eqref{GenKinEq-Lagrangian} in the form
\[
f_t(\boldsymbol\psi(t))\,\de^k\boldsymbol\psi(t)=f_0(\boldsymbol\psi_0)\,\de^k\boldsymbol\psi_0
\]
implies that the probability density evolves as a scalar function. That is, $f_t(\boldsymbol\psi(t))=f_0(\boldsymbol\psi_0)$ for Hamiltonian systems.
The invariant Liouville measure is the basic ingredient of modern ergodic theory.  {In the thermodynamics of closed Hamiltonian systems, the combination of reversibility and the existence of an invariant measure implies conservation of the entropy functional
\begin{equation}\label{GibbsEntropy}
S= \int\! f \log f \,\de^6\boldsymbol{z} 
\end{equation}
(in units of Boltzmann constant). 
At this point, one may inquire into the definition and dynamics of entropy. As we shall discuss  in Remark~\ref{rem:entropy} below, the concept of entropy for nonholonomic systems produces divergences that need to be treated by mathematical methods that are beyond the scope of the present paper. 

The issue of thermodynamic entropy may also be related to the closedness of rolling particle systems. 
For the special case of a body rolling on a substrate, nonholonomic dynamics emerges from the constraint force that the substrate exerts on each body. Thus, although this force does no work, according to the Lagrange-D'Alembert principle \cite{Bl2010}, the constraint still represents an interaction with the substrate. According to the ordinary thermodynamic definition of a closed system, the system would be still be defined as closed, since no energy transfer occurs from or into the substrate (at least in the absence of sliding). However, the point is that the conserved energy of  rolling bodies does \emph{not} imply Liouville volume conservation and thus ordinary thermodynamic considerations do not apply in general to a nonholonomic system.

\begin{remark}[{\color{black} Probability density and nonholonomic distributions}]\label{rmk-ProbDistr} As briefly mentioned in the previous section,
nonholonomic systems are defined on special geometric objects known as distributions. These are hyperplanes in phase-space to which the dynamics is confined. 
\rem{ 
In the general case, distributions possess \emph{no} differentiable structure and the concept of a probability density may not make sense in terms of a differential form of maximal degree. 
} 
These hyperplanes are generally not easy to handle. Our strategy in this paper will be to define a weak density on the whole phase space, that is \emph{supported} on the nonholonomic distribution. In this way, one can still use the usual tangent bundle structure of the entire phase space, while  particle dynamics is  localized on the distribution. However, any statements about the probability distribution will only hold in the weak sense.  That is, they will hold when integrated against smooth functions on phase space. 
\end{remark}
\rem{ 
{\color{magenta}\comment{CT: to me it sounds like second referee is right -- our distribution is regular and therefore it is a manifold. This has nothing to do with the difference between kinetic theory and single particle dynamics. I changed our sentences to simply say that putting coordinates on the constraint distribution would not be convenient, because this is not an ordinary phase space manifold. Therefore, we prefer to follow our delta function approach. Delete this box if you agree.}}
} 
\section{Lagrangian trajectories of rolling molecules}
\label{sec:E-L-rolling}
\subsection{Relabeling symmetry of collisionless kinetic theories}

It is clear from \eqref{L-to-E--map} that all the information contained in a collisionless kinetic equation is encoded in the phase-space transformation $\boldsymbol\psi(\boldsymbol\psi_0,t)$ taking the Lagrangian label $\boldsymbol\psi_0$ to its phase-space position at time $t$. Finding this transformation is equivalent to solving the equations of particle motion $\dot{\boldsymbol\psi}=\mathbf{X}(\boldsymbol\psi)$ with initial condition $\boldsymbol\psi(0)=\boldsymbol\psi_0$. {\color{black}Thus, one is normally interested in finding the vector field 
\[
\boldsymbol{X}=\dot{\boldsymbol\psi}\circ\boldsymbol\psi^{-1}\,,
\]
which produces the equations of motion. (Here, `$\circ$' denotes composition from the right.) The equations of motion for nonholonomic systems may be derived from Hamilton's variational principle \cite{Ho2008}}, provided one evaluates on constraint surfaces only \emph{after} taking variational derivatives. So, for non-interacting particles the vector field is well known and it can be derived from constrained Euler-Lagrange equations. When particles interact mutually with each other, the Klimontovich method \cite{Klimontovich}
may be applied to derive a collisionless kinetic equation. However, the application of the Klimontovich method on nonholonomic distributions may present difficulties, if the differentiable structure were to be lost, so that the divergence in \eqref{GenKinEq} would not make sense, {\color{black}see Remark \ref{rmk-ProbDistr}}. Even in the case when nonholonomic constraints in $\color{black}\Bbb{R}^6$ lead to an appropriate differential structure (e.g., as defined by the implicit function theorem), it can still be difficult to find a convenient coordinate system.

Thus, our strategy in formulating a consistent kinetic theory for nonholonomically constrained particles will be first to use Hamilton's variational principle to derive the equations for the Lagrangian trajectories $\boldsymbol\psi(\boldsymbol\psi_0,t)$ on the whole phase space and then to use reversibility to obtain the Eulerian vector field $\mathbf{X}(\boldsymbol{z})=\dot{\boldsymbol\psi}\circ\boldsymbol\psi^{-1}(\boldsymbol{z})$. The general formula for the variational principle reads
\begin{equation}\label{unreduced-Lagrangian-Gen}
\delta\int_{t_1}^{t_2}\int\mathscr{L}\big(\boldsymbol\psi(\boldsymbol\psi_0,t),\dot{\boldsymbol\psi}(\boldsymbol\psi_0,t)\big)\,f_0(\boldsymbol\psi_0)\,\de^k\boldsymbol\psi_0=0\,,
\end{equation}
which in turn yields constrained Euler-Lagrange equations through the nonholonomic constraint in expressing $\delta\boldsymbol\psi$.
The labels $\boldsymbol\psi_0$ are integrated over the  probability density $f_0(\boldsymbol\psi_0)\,\de^k\boldsymbol\psi_0$ whose support will later be confined to the nonholonomic distribution by using a Dirac delta function. This Lagrangian variational approach is widely used in the theory of inviscid fluid flows, and its application to kinetic equations was first due to Low \cite{Low1958}, who successfully showed that Vlasov-type equations  in plasma physics possess a Lagrangian {\color{black}variational} formulation. Forty years later, Low's work was revisited in \cite{CeHoHoMa1998}, within the modern mathematical language of {\color{black}Euler-Poincar\'e reduction by symmetry}. Since the action principle is invariant under relabeling, applying the inverse $\boldsymbol\psi^{-1}$ of $\boldsymbol\psi$ yields the following variational principle in terms of purely Eulerian variables:
\begin{align}\nonumber
0=&\
\delta\int_{t_1}^{t_2}\int\mathscr{L}\big(\boldsymbol\psi(\boldsymbol\psi_0,t),\dot{\boldsymbol\psi}(\boldsymbol\psi_0,t)\big)\,f_0(\boldsymbol\psi_0)\,\de^k\boldsymbol\psi_0
\\\nonumber
=&\
\delta\int_{t_1}^{t_2}\int\mathscr{L}\big(\dot{\boldsymbol\psi}\circ\boldsymbol\psi^{-1}(\boldsymbol{z}))\big)\,f(\boldsymbol{z})\,\de^k\boldsymbol{z}
\\
=&\!:
\delta\int_{t_1}^{t_2}\!\int \ell\big(\mathbf{X}(\boldsymbol{z})\big)\,f(\boldsymbol{z})\,\de^k\boldsymbol{z}
\label{GenEPLagrangian}
\end{align}
where $f$ is given by the above mentioned Lagrange-to-Euler map.  
 In the simplest case of holonomic particle dynamics, the variations
\begin{align*}
\delta\mathbf{X}&=\partial_t(\delta\boldsymbol\psi\circ\boldsymbol\psi^{-1})+\big((\delta\boldsymbol\psi\circ\boldsymbol\psi^{-1})\cdot\nabla\big)\mathbf{X}-\big(\mathbf{X}\cdot\nabla\big) (\delta\boldsymbol\psi\circ\boldsymbol\psi^{-1})
\\
\delta f&=-\operatorname{div}(f\,\delta\boldsymbol\psi\circ\boldsymbol\psi^{-1})
\end{align*}
yield the equations of motion
\begin{align}
&\frac{\partial}{\partial t}\frac{\delta l}{\delta \mathbf{X}}+(\mathbf{X}\cdot\nabla)\frac{\delta l}{\delta \mathbf{X}}+\frac{\delta l}{\delta \mathbf{X}}\operatorname{div}\mathbf{X}+\nabla\mathbf{X}\cdot\frac{\delta l}{\delta \mathbf{X}}=f\nabla\frac{\delta l}{\delta f}
\\
&\frac{\partial f}{\partial t}+\operatorname{div}\!\left(f\,\mathbf{X}\right)=0
\label{kineq}
\end{align}
 where $l(\mathbf{X},f):=\int\!\ell\big(\mathbf{X}(\boldsymbol{z})\big) f(\boldsymbol{z})\,\de^k\boldsymbol{z}$, so $\delta l/\delta{f} =\ell(\mathbf{X})$.
Although the equations appear to be coupled, the special nature of the Lagrangian that we shall choose will decouple the above system so that the second line acquires its own meaning as a kinetic equation.  The present paper extends this method to nonholonomic systems, in  considering the special example of interacting rolling balls.

\subsection{Configuration space and Lagrangian}

The main difficulty in developing the kinetic theory of nonholonomic systems can be explained as follows. 
Since Chaplygin's ball is a nonholonomic system, its dynamics (\emph{for an individual ball}) takes place on the \emph{distribution} $\mathcal{P}\subset  TSO(3)\times T\Bbb{R}^2 $ formed by the nonholonomic constraint 
\begin{equation}\label{distributionP}
\mathcal{P}=\left\{{(\bx,\bv,\R,v_\R)\in T( \Bbb{R}^2\times SO(3))}\,:\,\bv=v_\R\R^T( r \z+\R\boldsymbol\chi)\right\}
\,,
\end{equation}
{\color{black}where $TQ$ denotes the tangent bundle (i.e. the position-velocity phase space) associated to the configuration manifold $Q$, so that $T\Bbb{R}^k=\Bbb{R}^{2k}$ and $(\chi,\dot\chi)\in TSO(3)$ for any trajectory $\chi(t)\in SO(3)$.} 
In general, the distribution, (or set), defined by the nonholonomic constraint  does not possess the familiar tangent bundle structure common in mechanics \cite{Bl2010}. Thus,   special care must be taken to introduce any familiar concepts borrowed from calculus on  phase space manifolds. These difficulties persist if we wish to make a kinetic theory of nonholonomic systems. 
When we turn to the kinetic theory with non-holonomic constraints, the  dynamics is occurring on a subspace specified by a relation between velocities and coordinates, which is linear for the case of the gas of Chaplygin balls considered here.  One could, in principle, treat the dynamics on that  subspace using \emph{e.g} local coordinates, but we find it highly awkward. Instead, we propose a theory of generalized solutions that are defined in the whole phase space, but concentrated on the constraint subspace only.

Formally, the Lagrangian derivation of a collisionless kinetic theory should involve smooth invertible coordinate transformations $\color{black}\boldsymbol\psi$ (\emph{diffeomorphisms}) on $\mathcal{P}$, which in turn would determine the Lagrange-to-Euler map \eqref{L-to-E--map} associated to a probability density $\color{black}f$ on $\mathcal{P}$.
However, restricting the process to respect the geometric structure of the nonholonomic distribution $\mathcal{P}\subset T\Bbb{R}^2 \times TSO(3)$ leads to  several difficulties. We choose to circumvent some of these difficulties by considering dynamics on the full phase space $T\Bbb{R}^2 \times TSO(3)$ for as long as possible, before inserting the constraints using delta functions. 

Of course, physically, the dynamics makes no sense outside the distribution 
$\mathcal{P}$. Thus, in what follows, we consider kinetic {\color{black}probability densities} that are defined everywhere on $T\Bbb{R}^2 \times TSO(3)$, but are concentrated on the distribution $\mathcal{P}$ and vanish outside $\mathcal{P}$. To justify this assumption, we will show that kinetic densities concentrated on  $\mathcal{P}$ remain concentrated on $\mathcal{P}$ for all times under the evolutionary equations we shall derive. 
Thus, we define the 4-coordinate Lagrangian mapping (diffeomorphism) 
\begin{equation} 
\boldsymbol\psi:=\big( \psi_{\bx}, \psi_{\bv}, \psi_{\R}, \psi_{v_\R} \big) \in \Diff(T\Bbb{R}^2 \times TSO(3))
\label{psidef}
\end{equation}
which takes the initial coordinates $(\bx_0, \bv_0,\R_0, v_{\R0})$ to their values at the time $t$. Here, $\operatorname{Diff}(M)$ denotes the Lie group of diffeomorphisms of a manifold $M$. It may be worth emphasizing that, despite our notation, $\boldsymbol\psi$ is not simply a vector in 3D; rather it is a Lagrangian map on the phase space $T\Bbb{R}^2 \times TSO(3)$. In terms of this Lagrangian map, the rolling constraint identifies the following infinite-dimensional nonholonomic distribution in the ambient tangent bundle $T\Diff(T\Bbb{R}^2 \times TSO(3))$:
\begin{equation}
\mathcal{D}=\left\{T\Diff(T\Bbb{R}^2 \times TSO(3))\ |\ \dot{\psi}_\bx=\dot{\psi}_{\R}\!\left(r \psi_\R^T\,\z+\boldsymbol\chi\right)\right\}  .
\label{constr}
\end{equation}

In this section, we shall show how the following Lagrangian of the type \eqref{unreduced-Lagrangian-Gen} produces nonholonomic rolling dynamics:
\begin{align}
L_{f_0}=\frac12\int\! f_0(\bx_0,\R_0,\bv_0,v_{\R0})&\left(\left|\dot{\psi}_\bx\right|^2+\left\|\psi_{\R}^T\,\dot{\psi}_{\R}\right\|_J^2 - 2\, \psi_{\R}^T\,\bE \cdot \boldsymbol\chi 
\right. \nonumber 
\\
&\ 
\left.
+\left|\dot{\psi}_\bx-\psi_\bv\right|^2+\left\|{\psi}_\R^T\!\left(\dot{\psi}_\R-\psi_{v_\R}\right)\right\|_J ^2
\right) \de\bx_0\de\bv_0\de v_{\R0}\de\R_0 
,
\label{Lag-f0}
\end{align}
where
$\|\cdot\|_J^2$ denotes the norm given by the trace as
\[
\|A\|_ J^2:=\mathrm{Tr}\!\left(A^T J\, A\right) \, .  
\]
At this point the \emph{right trivialization map} \cite{HoScSt} gives the change of coordinates
\[
(\R,\nu):=(\R,v_\R\R^{-1})\color{black}\in SO(3)\times\mathfrak{so}(3)\,.
\]
where $\nu=v_\R\R^{-1}$ is an antisymmetric matrix belonging to the Lie algebra $\mathfrak{so}(3)$ of the rotation group $SO(3)$. {Notice that, for ease  of notation, we have dropped the hat symbol $\widehat{\ }$ usually accompanying antisymmetric matrices. See Remark \ref{rem-hatmap}.
We then have $\operatorname{Diff}(T\Bbb{R}^2 \times TSO(3))=\operatorname{Diff}(SO(3)\times\mathfrak{so}(3)\times T\Bbb{R}^2)$, so that the new Lagrangian reads
\begin{align}\label{EL-Lag}
L_{f_0}=\frac12\int\! f_0(\bx_0,\bv_0,\R_0,\nu_0)&\left(\left|\dot{\psi}_\bx\right|^2+\left\|\psi_{\R}^T\,\dot{\psi}_{\R}\right\|_J^2-2\bE \cdot\psi_{\R} \boldsymbol\chi
\right. \nonumber 
\\
&\ 
\left.
+\left|\dot{\psi}_\bx-\psi_\bv\right|^2+\left\|\dot{\psi}_\R\,\psi_\R^T-\psi_\nu\right\|_J ^2
\right)\de\bx_0\de\bv_0\de\nu_0\de\R_0
\end{align}
To summarize, we start with the tangent space $T\!\operatorname{Diff}(TQ)$ of diffeomorphisms of a tangent bundle $TQ=T(SO(3)\times\Bbb{R}^2)$ (not a distribution). Eventually, these  act on a probability density $f_0\in\operatorname{Den}(TQ)$ that is also defined on the same tangent bundle $TQ$. Then, enforcing the nonholonomic constraint $\dot\psi_\bx=\dot\psi_\R\!\left(r \psi_\R^T\,\z+\boldsymbol\chi\right)$ on the diffeomorphisms yields a dynamics that lies on  the constraint distribution $\mathcal{D}\subset T\!\operatorname{Diff}(TQ)$. So far, this is the standard symmetry reduction approach to nonholonomic systems \cite{BlKrMaMu1996}, although now it is being applied in infinite dimensions.

The key point of this paper comes from the fact that  the dynamics of the individual microscopic particles  take place on the constraint distribution $\mathcal{P}\subset TQ$, rather then the whole tangent bundle $TQ$. Then, 
\rem{ 
since $\mathcal{P}$ is not even a manifold, we cannot consider  such an object because there is no sense in which a probability density is defined on the constraint  distribution -- at least, we are not aware of such work and in any case, 
} 
the consideration of functions on $\mathcal{P}$ leads to several difficulties, mainly related to the fact that $\mathcal{P}$ is not an ordinary phase space manifold. Instead,  we avoid these difficulties by introducing a natural construction that is the key step in this paper. Namely,  we require that the probability density $f_0\in\operatorname{Den}(TQ)$ defined on the whole tangent bundle $TQ$ be \emph{supported} on the distribution $\mathcal{P}\subset TQ$. This is done by taking the following singular ansatz for the probability density 
\begin{equation}\label{singularPDF}
 f_0(\bx_0,\bv_0,\bnu_0,\R_0) = \phi_0(\bx_0,\bnu_0,\R_0) \delta(\bv_0-\bnu_0 \times \bsigma(\R_0)),
\end{equation}
then showing that the dynamics of the resulting kinetic theory preserves this class of solutions.
As we shall see, this singular ansatz for the probability density leads to Euler-Lagrange equations involving the whole tangent bundle $TQ$, which are \emph{supported} on the nonholonomic distribution $\mathcal{P}\subset TQ$ on which the particles are constrained to move. This is a natural picture, since the dynamics does not ``see'' what is outside the distribution $\mathcal{P}$.
In this way, we avoid the difficulties of dealing with the densities on $\mathcal{P}$ only and, as we show below, derivation of kinetic theory is possible. We believe that a similar approach can be generalized for an arbitrary system of nonholonomically constrained particles, without many substantial difficulties. Notice that the probability density $\phi_0=\int f_0\,\de^3\bv_0$ is \emph{not} a density on the distribution $\mathcal{P}$; rather, this is more simply a probability density on $\color{black}\Bbb{R}^2\times TSO(3) \simeq \Bbb{R}^2\times SO(3)\times\mathfrak{so}(3)$.  Notice that \eqref{singularPDF} may involve all the complications coming by the fact that we have to work with measures instead of probability densities; we refer the reader to \cite{BoCa}, for an example of a complete analysis in the case of constrained biology kinetic models.


\subsection{Hamilton's principle and Euler-Lagrange equations}

The Euler Lagrange equations arise from Hamilton's principle
\[
\delta\int_{t_1}^{t_2}L_{f_0}(\boldsymbol\psi,\dot{\boldsymbol\psi})\,\de t=0
\,.
\]
The variations are subject to the constraint
\begin{equation}\label{nh-constraint}
\delta{\psi}_\bx=\delta{\psi}_{\R}\psi_\R^T\,\bsigma(\psi_\R)
\end{equation}
Then, upon denoting the Liouville volume by $\de w_0=\de\bx_0\de\bv_0\de\nu_0\de\R_0$, one computes the following relation
\begin{align*}
\int_{t_1}^{t_2}\!\!\int\!\left(\frac{\delta L_{f_0}}{\delta \dot{\psi}_\bx}\cdot\delta\dot{\psi}_\bx+\frac{\delta L_{f_0}}{\delta {\psi}_\bx}\cdot\delta{\psi}_\bx\right)\!\de w_0\,\de t
=
&
-\int_{t_1}^{t_2}\!\!\int\!\left(\frac{\de}{\de t}\frac{\delta L_{f_0}}{\delta \dot{\psi}_\bx}-\frac{\delta L_{f_0}}{\delta {\psi}_\bx}\right)\!\cdot\delta{\psi}_\bx\, \de w_0\,\de t
\\
=
&
-\int_{t_1}^{t_2}\!\!\int\!\left(\frac{\de}{\de t}\frac{\delta L_{f_0}}{\delta \dot{\psi}_\bx}-\frac{\delta L_{f_0}}{\delta {\psi}_\bx}\right)\!\cdot\!\left(\delta{\psi}_{\R}\psi_\R^T\,\bsigma(\psi_\R)\right)\de w_0\,\de t
\\
=
&
-\int_{t_1}^{t_2}\!\!\int\!\operatorname{Tr}\!\left(\!\left(\frac{\de}{\de t}\frac{\delta L_{f_0}}{\delta \dot{\psi}_\bx}-\frac{\delta L_{f_0}}{\delta {\psi}_\bx}\right)\!\bsigma^T\!\left(\delta\psi_{\R}\,\psi_{\R}^T\right)^{\!T}\right)\!\de w_0\,\de t
\\
=
&
-\int_{t_1}^{t_2}\!\!\int\!\operatorname{Tr}\!
\left(\!\left(\frac{\de}{\de t}\frac{\delta L_{f_0}}{\delta \dot{\psi}_\bx}\,\bsigma^T-\frac{\delta L_{f_0}}{\delta {\psi}_\bx}\,\bsigma^{T\!}\right)^{\!\!\mathcal{A}\!\!}\psi_\R\,\delta\psi_{\R}^T\right)\!\de w_0\,\de t
\end{align*}
where the superscript $(\,\cdot\,)^\mathcal{A}$ denotes {\bfi antisymmetric part}.
Next, Hamilton's principle yields the Euler-Lagrange equations
\begin{align}
&
\left(\!\left(\frac{\de}{\de t}\frac{\delta L_{f_0}}{\delta \dot{\psi}_\R}-\frac{\delta L_{f_0}}{\delta {\psi}_\R}\right)\!\psi_{\R}^T
+\left(\frac{\de}{\de t}\frac{\delta L_{f_0}}{\delta \dot{\psi}_\bx}
-
\frac{\delta L_{f_0}}{\delta \psi_\bx}\right)\!\bsigma^T\!\right)^{\!\mathcal{A}\!}=0
\,,\nonumber 
\\
&
\frac{\de}{\de t}\frac{\delta L_{f_0}}{\delta \dot{\psi}_\nu}-\frac{\delta L_{f_0}}{\delta {\psi}_\nu}=0
\,,\label{NH-EL}
\\
& 
\nonumber 
\frac{\de}{\de t}\frac{\delta L_{f_0}}{\delta \dot{\psi}_\bv}-\frac{\delta L_{f_0}}{\delta {\psi}_\bv}=0
\,.
\end{align}
The last two Euler-Lagrange equations lead to the relationship{s} 
\[
f_0\left(\dot{\psi}_\bx-\psi_\bv\right)=0
\quad\hbox{and}\quad
f_0\left(\dot{\psi}_\R-\psi_\nu\psi_\R\right)=0
\,,
\]
where we keep in mind the singular expression of $f_0$ given in \eqref{singularPDF}, 
so we do not divide out by $f_0$.  Then, upon introducing $\bsigma(\psi_\R)=r\z+\psi_\R\boldsymbol\chi$, the nonholonomic constraint $\dot{\psi}_\bx=\dot{\psi}_\R\,\psi^T_\R\bsigma(\psi_\R)$ yields
\begin{equation}\label{constraint}
f_0\left(\psi_\bv-\psi_{\bnu}\times\bsigma(\psi_\R)\right)=0
\,.
\end{equation}
Notice that the presence of a delta function in the density $f_0$ forbids dividing by $f_0$.
{\color{black}
\begin{theorem}[Nonholonomic Euler-Lagrange equations]\label{ELThm} Subject to the constraint \eqref{nh-constraint},
Hamilton's variational principle associated to the Lagrangian functional \eqref{EL-Lag} yields 
\[
f_0\!\left(\dot{\psi}_\bx-\psi_{\bnu}\times\bsigma(\psi_\R)\right)=0
\quad\hbox{and}\quad
f_0\!\left(\dot{\psi}_\R-\psi_\nu\psi_\R\right)=0
\,,
\]
Also, upon introducing the director $\bn(\psi_\R)=\psi_\R\bchi$ and the microinertia tensor $j(\psi_\R)=\psi_\R J\psi_\R^T$, the resulting nonholonomic Euler-Lagrange equation reads
\begin{equation}\label{NH-EL-equation}
f_0\!\left(j(\psi_\R)\,\dot{\psi}_{\bnu}-j(\psi_\R) \psi_{\bnu}\times{\psi}_{\bnu}\right)
=
f_0\!\left(\bE\times\bn(\psi_\R)
-
\bsigma(\psi_\R)\times\bsigma(\psi_\R)\times\dot{\psi}_{\bnu}-\bsigma(\psi_\R)\times\dot\bsigma(\psi_\R)\times\psi_{\bnu}\right)
.
\end{equation}
\end{theorem}
\paragraph{Proof.} The first part of the theorem has already been proven above. We only need to prove the second part. After computing
\begin{align*} 
&
\frac{\delta L_{f_0}}{\delta \dot{\psi}_\R}=f_0\,j\, \psi_\nu\psi_\R \, , 
\\
&
\frac{\delta L_{f_0}}{\delta {\psi}_\R}=f_0\left(\!(\bE\boldsymbol\chi^T\psi_\R^T)^{\mathcal{A}}+\frac12\left[j,\psi_\nu\psi_\nu\right]\right)\!\psi_\R \, , 
\\
&\frac{\delta L_{f_0}}{\delta \dot{\psi}_\bx}=f_0\,{\psi}_\bv
\,,
\\
&
\frac{\delta L_{f_0}}{\delta {\psi}_\bx}=0 \, , 
\end{align*}
the remaining equation gives the dynamics of the local angular momentum
\[
\Pi:=\frac{\delta L_{f_0}}{\delta \dot{\psi}_\R}\,\psi^T_\R\in\mathfrak{so}^*(3)\,.
\]
Indeed,  pairing the first equation in \eqref{NH-EL} against $\delta\psi_\R$ and differentiating by parts yields
\begin{align*}
\frac{\de}{\de t}\!\left(\dede{L_{f_0}}{\dot{\psi}_\R}\,\psi_\R^{-1}\right)^{\!\mathcal{A}}
&=
-
\left(\dede{L_{f_0}}{\dot{\psi}_\R}\,\psi^{-1}_\R\,\dot{\psi}_\R\,\psi^{-1}_\R\!\right)^{\!\mathcal{A}}
+
\left(\!\left(\frac{\de}{\de t}\dede{L_{f_0}}{\dot{\psi}_\R}\right)\!\psi_\R^{-1}\!\right)^{\!\mathcal{A}}
\\
&=\underbrace{
-\left(\dede{L_{f_0}}{\dot{\psi}_\R}\,\psi^{-1}_\R\,{\psi}_\nu\right)^{\!\mathcal{A}\!}
+
\left(\dede{L_{f_0}}{{\psi}_\R}\,\psi_\R^{-1}\right)^{\!\mathcal{A}}
}_\text{micropolar terms}
-
\underbrace{
\left(\!\left(\frac{\de}{\de t}\frac{\delta L_{f_0}}{\delta \dot{\psi}_\bx}
-
\frac{\delta L_{f_0}}{\delta \psi_\bx}\right)\!\bsigma^T\!\right)^{\!\mathcal{A}\!}
}_\text{nonholonomic terms} \, . 
\end{align*}
Therefore, the desired equation of motion is 
\begin{align*}
f_0\,\frac{\de}{\de t}\!\left(j(\psi_\R) \psi_\nu\right)^{\!\mathcal{A}}
=&\,
-f_0\left(j(\psi_\R) \psi_\nu\,{\psi}_\nu\right)^{\mathcal{A}}
+
f_0\!\left(\bE\boldsymbol\chi^T\psi_\R^T+\frac12\left[j,\psi_\nu\psi_\nu\right]\right)^{\!\mathcal{A}}
-
f_0\left(\psi_\bv\,\bsigma^T\right)^{\mathcal{A}}
\\
=&\,
f_0\left(\bE\,\bn^{T\!}(\psi_\R)\right)^{\mathcal{A}}
-
f_0\big(\dot{\psi}_\bv\,\bsigma^T\big)^{\mathcal{A}}
\,.
\end{align*}
Since $\de{j(\psi_\R)}/\de t=\left[\psi_\nu,j\right]$, we then find \eqref{NH-EL-equation}, upon using the inverse of the hat map and its properties (see e.g. \cite{HoScSt}). $\blacksquare$

\bigskip

\noindent
Notice that}, upon dividing \eqref{NH-EL-equation} by the $\phi_0(\bx_0,\bnu_0,\R_0)$ in \eqref{singularPDF} and integrating over $\bv_0$, we obtain
\begin{equation*}
\left.\left(j\dot{\psi}_{\bnu}-j \psi_{\bnu}\times{\psi}_{\bnu}- {\bE}\times\bn-
\bsigma\times\bsigma\times\dot{\psi}_{\bnu}-\bsigma\times\dot\bsigma\times\psi_{\bnu}\right)\right|_{\bv_0=\bnu_0\times\bsigma(\R_0)}=0
\,.\end{equation*}
Consequently, Lagrangian dynamics is supported on the distribution \eqref{distributionP}, consistently with the single molecule dynamics. Moreover, this calculation shows that the ansatz (\ref{singularPDF}) remains consistent for all times.

\section{The Euler-Poincar\'e approach to kinetic theory}
\label{sec:E-P-rolling}
As an alternative to the Euler-Lagrange formulation, let us derive these equations of motion using the Euler-Poincar\'e theory, which explicitly utilizes the symmetry reduction in the Lagrangian.   
This reduction arises from the relabeling symmetry introduced in equation \eqref{GenEPLagrangian}. 

\subsection{Nonholonomic Euler-Poincar\'e equations of motion}
This section makes use of the following invariance relation (see \eqref{GenEPLagrangian}):
\[
L_{f_0}(\dot{\boldsymbol\psi},{\boldsymbol\psi})=L_{{\boldsymbol\psi}* f_0}(\dot{\boldsymbol\psi}\circ{\boldsymbol\psi}^{-1})=:l(X,f)
\]
where $X=\dot{\boldsymbol\psi}\circ{\boldsymbol\psi}^{-1}$ is the vector field transporting the probability density $f={\boldsymbol\psi}_* f_0$, where $\psi_* f_0$ is the push-forward given by the Lagrange-to-Euler map \eqref{L-to-E--map}
\begin{align}\nonumber
f=&\int\!\!f_0(w_0)\,\delta(\bx-\psi_\bx(w_0))\,\delta(\bv-\psi_\bv(w_0))\,\delta(\R-\psi_\R(w_0))\,\delta(\nu-\psi_\nu(w_0)) \mbox{d} w_0 
\\\label{NH-L-to-E-MAP}
=&
\int\!\!\phi_0(a_0)\,\delta(\bx-\widetilde{\psi}_\bx(a_0))\,\delta(\R-\widetilde{\psi}_\R(a_0))\,\delta(\nu-\widetilde{\psi}_\nu(a_0))\,\delta(\bv-\widetilde{\psi}_{\bv}(a_0))\, \de a_0 
\,.
\end{align}
Here we have introduced the notation 
$\widetilde{\psi}(a_0):=\left.\psi(w_0)\right|_{\bv_0=\bnu_0\times\bsigma(\R_0)}$ and
\[
w_0:=(\bx_0,\bv_0,\R_0,\nu_0)
\,,\qquad
a_0:=(\bx_0,\R_0,\nu_0)
.
\]
Taking the time derivative of the Lagrange-to-Euler map ({\color{black}and pairing it with a test function}) produces the evolution equation for the probability density
\begin{equation}\label{EP-Vlasov}
\frac{\partial f}{\partial t}+\nabla\cdot(f X)=0
\end{equation}
{\color{black} 
where $X=X(\bx,\bv,\bnu,\R)=\dot{\boldsymbol\psi}\circ{\boldsymbol\psi}^{-1}$.

Upon using the Lie derivatve notation $\pounds_X$ (see \cite{HoMaRa1998}) and by introducing $\eta=\delta{\boldsymbol\psi}\circ{\boldsymbol\psi}^{-1}$,  the variation $\delta X=\dot{\eta}+[X,\eta]$ 
yields
\begin{align*}
\delta\!\int_{t_1}^{t_2}\! l(X,f)\,\de t
=&
\int_{t_1}^{t_2} \!
\left\langle -\frac{\partial}{\partial t}\dede{l}{X}-\pounds_{X}\dede{l}{X}+f\nabla\dede{l}{f},\eta\right\rangle
\de t
\,,
\end{align*}
where  $\delta\ell/\delta X$ is a differential one-form density on $\Bbb{R}^3\times\mathfrak{so}(3)\times SO(3)$ (see \cite{HoMaRa1998} for the explicit forms of the Lie derivative)  and $\langle\cdot,\cdot\rangle$ denotes the pairing between vector fields and one-form densities on the same space.
Then, we make use of the constraint
\[
\eta_\bx=\eta_\R\R^{T}\bsigma(\R)
\,,
\]   
and} the Euler-Poincar\'e equations read
\begin{align}\label{EP-EQ-1}
&\Bigg(\underbrace{\left(\frac{\partial}{\partial t}\dede{l}{X}+\pounds_{X}\dede{l}{X}-f\nabla\frac{\delta l}{\delta f}\right)_{\!\R\!}\!\R^T}_\text{ordinary terms}
+\,
\underbrace{\left(\frac{\partial}{\partial t}\dede{l}{X}+\pounds_{X}\dede{l}{X}-f\nabla\frac{\delta l}{\delta f}\right)_{\!\!\bx}\bsigma^T}_\text{nonholonomic terms}\Bigg)^{\!\mathcal{A}\!}
=
0
\\\label{EP-EQ-2}
&
\left(\frac{\partial}{\partial t}\dede{l}{X}+\pounds_{X}\dede{l}{X}-f\nabla\frac{\delta l}{\delta f}\right)_{\!\nu}=0
\,,\qquad
\left(\frac{\partial}{\partial t}\dede{l}{X}+\pounds_{X}\dede{l}{X}-f\nabla\frac{\delta l}{\delta f}\right)_{\!\bv}=0
\end{align}
Following \cite{KiPu2010}, in order to account for collective interactions among molecules, such as dipole and Lennard-Jones, we also  introduce an appropriate interaction potential  $U=U(\bx-\bx',\R'\R^T)$, which generates a term
\[
U*\!\int \!f\,\de\bv\de\bnu=\int U(\bx-\bx',\R'\R^T)\int\!f(\bx',\bv,\bnu,\R')\,\de\bv\de\bnu\,\de\bx'\de\R'
\]
in the Lagrangian. 
Then, the Euler-Poincar\'e Lagrangian reads
\begin{multline}\label{EP-Lag}
l=\frac12\int\!f\bigg(\!\left|u_\bx\right|^2+\left\|\R^T u_{\R}\right\|_J^2-2\,\R^T\, \bE\cdot \boldsymbol\chi- 2U*\!\int\!f\,\de\bv'\de\bnu' 
\\ 
+\left|u_\bx-\bv\right|^2+\left\| u_\R\R^T-\nu\right\|_J ^{2\!}
\bigg) \de\bx\de\bv\de\bnu\de\R 
\,.
\end{multline}
The other two equations of motion yield
\begin{equation}
f\left(u_\R-\nu\R\right)=0
\,,\qquad
f\left(u_\bx-\bv\right)=0
\,,\label{feq12}
\end{equation} 
so that finally
\[
f X=f\left(\bv,\nu\R,a_{\nu},a_{\bv}\right)
\,.\]
Notice that upon multiplying the nonholonomic constraint $u_\bx=u_\R\,\R^T\bsigma(\R)$ by $f$, the equations above would yield
\begin{equation}\label{TrivialIdentity}
f\left(\bv-\bnu\times\bsigma(\R)\right)\equiv0\,.
\end{equation}
If $f$ were a smooth nonvanishing function, this relation would be contradictory, because it would imply a relation between independent coordinates. However, the singular expression for $\color{black}f_0$ in equation (\ref{singularPDF}) will be seen to avoid this difficulty, as a result of the following Lemma, {\color{black}which will be proven in two different ways}.
\begin{lemma}[Singular probability density function]\label{lem:delta}
If $f_0$ is given by \eqref{singularPDF}, the definition $f=\psi_* f_0$ and the relation \eqref{constraint}  imply
\begin{equation}\label{delta-f}
f(\bx,\bv,\bnu,\R,t)=\phi(\bx,\bnu,\R,t)\,\delta(\bv-\bnu\times\bsigma(\R))
\end{equation}
where
\begin{equation}\label{phi-def}
\phi(\bx,\bnu,\R)=(\tilde\psi_*\phi_0)(\bx,\bnu,\R)=\int\!\phi_0(a_0)\,\delta(\bx-\tilde\psi_\bx(a_0))\,\delta(\R-\tilde\psi_\R(a_0))\,\delta(\nu-\tilde\psi_\nu(a_0))\,\mbox{d} a_0\,.
\end{equation}
\end{lemma}
\paragraph{Proof.}  
Recall the relation \eqref{constraint} $f_0 \big( \psi_{\bv}-\psi_{\bnu} \times \bsigma \big)=0$. Integrating this identity over $\bv_0$ yields
\[
\tilde\psi_\bv(a_0))=\tilde\psi_{\bnu}(a_0))\times\bsigma(\tilde\psi_\R(a_0))
\,.
\]
Then, the Lagrange-to-Euler map \eqref{NH-L-to-E-MAP} becomes
\begin{align}\nonumber
f=&
\int\!\!\phi_0(a_0)\,\delta(\bx-\widetilde{\psi}_\bx(a_0))\,\delta(\R-\widetilde{\psi}_\R(a_0))\,\delta(\bnu-\widetilde{\psi}_{\bnu}(a_0))\,\delta \big (\bv-\widetilde{\psi}_{\bnu}(a_0) \times\bsigma(\widetilde{\psi}_\R(a_0)) \big)\, \de a_0 
\end{align}
and thus taking the pairing with an arbitrary test function $h(\bx,\bv,\bnu,\R)$, we have
\begin{align*}
\int \!h(\bx,\bv,\bnu,\R)\, f(\bx,\bv,\bnu,\R,t)&\, \de\bx \de\bv\de\nu\de\R=
\\
=&
\int \!h\big(\widetilde{\psi}_\bx(a_0),\widetilde{\psi}_\nu(a_0))\times\bsigma(\widetilde{\psi}_\R(a_0)),\widetilde{\psi}_\nu(a_0),\widetilde{\psi}_\R(a_0)\big)\,\phi_0(a_0)\,\de a_0 
\\
=&
\int\!\de\bx \de\nu\de\R\int\!\de\bv\, h(\bx,\bv,\nu,\R)\,\phi(\bx,\nu,\R)\,\delta(\bv-\bnu\times\bsigma(\R))
\\
=&
\int\!\de\bv\int \!\de\bx\de\nu\de\R\, h(\bx,\bv,\nu,\R)\,\phi(\bx,\nu,\R)\,\delta(\bv-\bnu\times\bsigma(\R))
\\
=& 
\int\! h(\bx,\bv,\bnu,\R)\, \phi(\bx,\nu,\R)\,\delta(\bv-\bnu\times\bsigma(\R))\,\de\bx \de\bv\de\nu\de\R
,
\end{align*}
where the intermediate steps follow by direct verification.
Consequently,
\[
\int\! h(\bx,\bv,\bnu,\R)\Big(f(\bx,\bv,\bnu,\R)-\phi(\bx,\nu,\R)\,\delta(\bv-\bnu\times\bsigma(\R))
\Big)\de\bx \de\bv\de\nu\de\R=0
\]
and the statement of the Lemma follows, since $h$ is arbitrary. $\blacksquare$\\

\begin{remark}
Thus, the potentially troublesome equation (\ref{TrivialIdentity}) presents no difficulty, as it is satisfied automatically because the kinetic probability distribution $f$ takes the form (\ref{delta-f}). Namely, because $f$ takes the form (\ref{delta-f}), equation (\ref{TrivialIdentity}) follows from the relation for delta functions that one interprets $x\delta(x)=0$ for argument $x$.  Notice that this result follows easily from the invariance of the constraint $\bv=\bnu\times\bsigma(\R)$.   Indeed, 
is not surprising since this advection field has been constructed by minimizing the Lagrangian in such a way that the flow leaves invariant the constraints.  
 In the next section, we shall show an alternative proof of this Lemma that utilizes the dynamics of the {\color{black}kinetic} equation. The interpretation of equation (\ref{feq12}) will be discussed further in Remark \ref{paradoxresolved}. 
\end{remark}
{At this point, one may proceed by deriving the Euler-Poincar\'e equations from the following theorem, whose proof may be found in Appendix \ref{App:proof1}.
\begin{theorem}\label{theorem-f} Upon using the reduced Lagrangian \eqref{EP-Lag}, the nonholonomic Euler-Poincar\'e equations \eqref{EP-EQ-2} produce
\[
f\left(u_\R-\nu\R\right)=0
\,,\qquad
f\left(u_\bx-\bnu\times\bsigma\right)=0
\,.
\]
Also, define the director $\bn(\R)=\R\boldsymbol\chi$ and the microinertia tensor $j(\R)=\R J\R^{-1}$. Then, upon recalling the notation 
$\overrightarrow{A}_i:=\epsilon_{ijk}\,A_{jk}$
for an arbitrary antisymmetric matrix $A$, the Euler-Poincar\'e equation \eqref{EP-EQ-1} yields
\begin{multline}
f\left(ja_{\bnu}
+\bsigma\times\bsigma\times a_{\bnu}\right)
\\
=
f\!\left(j\bnu\times\bnu
+\bE\times\bn-\overrightarrow{\left(\partial_\R U* \!\int \!f\,\de\bv\de\bnu \right)\R^T}
+\bsigma\times(\bnu\times\bnu\times\bn)
+\bsigma\times\partial_\bx U*f\right)
\,.
\end{multline}
\end{theorem}

\smallskip

\noindent
Therefore, the equation for the $\nu$-component $a_{\bnu}$ of the Euler-Poincar\'e vector field $X$ reads
\begin{equation}\label{EP-EQ-nu-component}
f a_{\bnu}=f\left(j+\widehat\sigma\widehat\sigma\right)^{-1}\left(j\bnu\times\bnu
+\bE\times\bn-\overrightarrow{\left(\partial_\R U* \!\int \!f\,\de\bv\de\bnu \right)\R^T}
+\bsigma\times(\bnu\times\bnu\times\bn)+\bsigma\times\partial_\bx U*f
\right)
\end{equation}
where $\widehat{\sigma}\widehat{\sigma}=\bsigma\bsigma^T-\sigma^2\mathbb{I}$ is a traceless symmetric matrix that is produced by the nonholonomic constraint and modifies the microinertia  \makebox{tensor $j$}.

\subsection{The probability density function}
Upon recalling the Lagrangian dynamics arising from Euler-Lagrange equations, the Lagrange-to-Euler map {\color{black}$f=\psi_*f_0$ in \eqref{L-to-E--map}}
 can be applied to give the explicit expression of the probability density, as was done in the previous section. Then, the time derivative of the Lagrange-to-Euler map {(\color{black}appropriately paired with a test function) yields the kinetic equation \eqref{EP-Vlasov} with $X=X(u_\bx,u_\R,a_{\bnu},a_{\bv})$. This kinetic equation becomes}
\begin{equation}\label{PrimitiveVlasov}
\frac{\partial f}{\partial t}+\nabla_\bx\cdot(f u_\bx)+\nabla_\bv\cdot(f a_\bv)+\nabla_{\bnu}\cdot(f a_{\bnu})+\nabla_\R\cdot(f u_\R)=0
\end{equation}
where the relations
\begin{align*}
f u_\R=
f\nu\R
\quad\hbox{and}\quad
f u_\bx=
f\bnu\times\bsigma(\R)
\end{align*}
were found in the previous section, along with {\color{black} the expression \eqref{EP-EQ-nu-component} for} $f a_{\bnu}$. At this point, upon recalling the relations \eqref{Push-forward} and \eqref{constraint},  we observe that $f a_\bv$ can be computed from the Euler-Lagrange equations as follows
\begin{align}
f a_\bv
=f \dot\psi_\bv\circ\psi^{-1}
=&\,
f\left(\dot\psi_\nu\bsigma(\psi_\R)+\psi_\nu\dot{\bsigma}(\psi_\R)\right)\psi^{-1}
\nonumber 
\\
=&\,
f\dot\psi_\nu\psi^{-1}\bsigma(\R)+f\nu\left(\dot\psi_\R\psi^{-1}\cdot\partial_{\R\!}\right)\!\bsigma(\R)
\nonumber 
\\
=&\,
f a_\nu\bsigma(\R)+f\nu\left(u_\R\cdot\partial_{\R\!}\right)\!\bsigma(\R)
\nonumber 
\\
=&\,
f a_{\bnu}\times\bsigma(\R)+f \bnu\times\bnu\times\bn
\label{av}
.
\end{align}
Consequently, the above expression for $a_\bv$ allows to write the kinetic equation \eqref{PrimitiveVlasov} as
\begin{equation}
\frac{\partial f}{\partial t}
+
\bnu\times\bsigma\cdot\pp{f}{\bx}
 -
\operatorname{Tr}\!\left(\!\R^{T\,}\widehat{\nu}\,\pp{f}{\R}\right) 
+
\pp{}{\boldsymbol\nu}\cdot\left(f\,a_{\boldsymbol\nu}\right)
+
\pp{}{\bv}\cdot(f\,a_\bv)=0 \, . 
\label{EPVlasov2}
\end{equation}
This result allows us to demonstrate an alternative proof of Lemma~\ref{lem:delta}, by writing the evolution equation in the weak form. 
\begin{lemma}[Solutions on the distribution]
Suppose the initial condition for the probability function $f$ are of the following form 
\begin{equation} 
f(t=0,\bx,\bnu,\R,\bv)=\phi_0(\bx,\bnu,\R) \delta \big(\bv-\bnu \times \bsigma(\R) \big) \, , 
\label{f0assumption}
\end{equation} 
 Then, at arbitrary $t>0$, 
\begin{equation} 
f(t,\bx,\bnu,\R,\bv)=\phi(t,\bx,\bnu,\R) \delta\big(\bv-\bnu \times \bsigma(\R) \big) \, , 
\label{f0sol}
\end{equation}
\emph{i.e.}, the concentrated solution preserves its form under evolution. Moreover, evolution of $\phi$ is given by the equation 
\begin{equation}
\frac{\partial \phi}{\partial t}
+
\bnu\times\bsigma\cdot\pp{\phi}{\bx}
 -
\operatorname{Tr}\!\left(\!\R^{T\,}\widehat{\nu}\,\pp{\phi}{\R}\right) 
+
\pp{}{\boldsymbol\nu}\cdot\left(\phi\,a_{\boldsymbol\nu}\right)=0 \,.
\label{EPVlasov-final}
\end{equation}
\end{lemma}
\noindent
{\bf Proof}. Consider, for now, $f = \phi(\bx, \bnu, \R, t) g(\bv - \bnu \times \bsigma(\R) )$, where $g$ is an arbitrary generalized function. Then, for an arbitrary test function $\zeta(\bv)$, we have 
\begin{align} 
\Big< g &\Big[ 
\frac{\partial \phi}{\partial t}
+
\bnu\times\bsigma\cdot\pp{\phi}{\bx}
 -
\operatorname{Tr}\!\left(\!\R^{T\,}\widehat{\nu}\,\pp{\phi}{\R}\right)
+
\pp{}{\boldsymbol\nu}\cdot\left(\phi\,a_{\boldsymbol\nu}\right) \Big]
\, , \, {\zeta}(\bv) \Big>_{\bv} 
\label{phieqweak}
\\
&\quad + 
\Big< \phi \Big[ 
\hat{\nu} R \cdot  \pp{g}{R} 
+
 a_{\bnu} \cdot \pp{g}{\bnu} 
+
a_{\bv} \cdot \pp{g}{\bv} 
\Big]
\, , \, 
{\zeta}(\bv) \Big>_{\bv} 
 =0 \, . 
\label{geqweak}
\end{align} 
 The pairing, denoted by $< , >_{\bv}$, is assumed to be the usual $L^2$ pairing \emph{in $\bv$ coordinate only}. Then, since only $g$ depends on $\bv$, and $g=\delta(\bv - \bnu \times \bsigma)$, we can 
rewrite (\ref{phieqweak},\ref{geqweak}) as 
\begin{align} 
&  \Big[ 
\frac{\partial \phi}{\partial t}
+
\bnu\times\bsigma\cdot\pp{\phi}{\bx}
 -
\operatorname{Tr}\!\left(\!\R^{T\,}\widehat{\nu}\,\pp{\phi}{\R}\right) 
+
\pp{}{\boldsymbol\nu}\cdot\left(\phi\,a_{\boldsymbol\nu}\right) \Big]
\cdot  {\zeta}(\bnu \times \bsigma )
\label{phieqweak2}
\\
 &
-  
\Big[ 
a_{\bv} - a_{\bnu} \times \bsigma- \bnu \times \bnu \times \bn
\Big] \cdot \pp{{\zeta}}{\bv}(\bnu \times \bsigma) =0   \, . 
\label{geqweak2}
\end{align} 
From the definition of $\bsigma=\hat{\bz}+\R \boldsymbol{\chi}=\hat{\bz}+ \bn$, with $\boldsymbol{\chi}$ being a constant vector denoting the distance from the geometric center to the center of mass in the body frame, we get
\[ 
 \pp{\sigma}{\R} \R^T= \big( \R \boldsymbol{\chi} \big)^T= \bn^T. 
\] 
The expression in square brackets in (\ref{geqweak2}) was obtained using the following useful 
identities for the derivatives of $\delta(\bxi)$: 
\[
\Bigg<   {\rm Tr}\Big(\hat{\nu}   \pp{\delta}{\R} \R^T\Big) \, , \, \eta (\bv)  \Bigg>_{\!\bv}=  
\Bigg< {\rm Tr}\Big(  \hat{\nu}   \pp{\delta}{\bxi} \hat{\nu} \bn \Big) 
\, , \, {\zeta}(\bv) \Bigg> _{\!\bv}
=  \big( - \bnu \times \bnu \times \bn\big) \cdot  \pp{{\zeta}}{\bv} (\bnu \times \bsigma)\, .    
\]
Also, 
\[ 
\Bigg<   a_{\bnu} \cdot \pp{\delta}{\bnu}   \, , \, {\zeta}(\bv) \Bigg>_{\!\bv}=  
- \big( a_{\bnu} \times \bsigma \big) \cdot  \pp{{\zeta}}{\bv}(\bnu \times \bsigma) 
\] 
and 
\[ 
\Bigg<   a_{\bv} \cdot \pp{\delta}{\bv}  \, , \, {\zeta}(\bv) \Bigg>_{\!\bv}
=
  a_{\bv} \cdot \pp{{\zeta}}{\bv}(\bnu \times \bsigma)  \,  \, .  
\] 
Therefore, the equation in square brackets in (\ref{geqweak2}) is equal to 
\[ 
\big( a_{\bv} - a_{\bnu} \times \bsigma- \bnu \times \bnu \times  \bn \Big) \cdot  \pp{{\zeta}}{\bv}(\bnu \times \bsigma)  =0 \, , 
\] 
according to (\ref{av}).

From the coefficient of ${\zeta}(\bnu \times \bsigma)$, we see that the equation (\ref{EPVlasov-final}) is valid. The coefficient of $\nabla_{\bv}\zeta$, evaluated at the point $\bnu \times \bsigma$, vanishes identically due to (\ref{av}). Thus, the solution (\ref{f0sol}) remains concentrated on the constraint distribution $\bv = \bnu \times \bsigma(\R)$ for all times. $\blacksquare$

\begin{remark}\label{paradoxresolved}
A similar calculation performed in the strong sense, without multiplication by a test function ${\zeta}(\bv)$, proves that the structure of any solution of the form 
\begin{equation}  
f= \phi(\bx, \bnu, \R, t) g(\bv - \bnu \times \bsigma) 
\label{fgansatz}
\end{equation}
is preserved under the evolution for an arbitrary smooth function of $g$.   Again, this result follows easily from the invariance of the constraint $\bv=\bnu\times\bsigma(\R)$. However, we cannot assign any physical meaning to this interesting mathematical fact, as the solutions of the type (\ref{fgansatz}) do not concentrate the probability function $f$ onto the constraint manifold. Alternatively, if one were to take $g(\boldsymbol\xi)\neq\delta(\boldsymbol\xi)$ in (\ref{fgansatz}) above, then relation \eqref{TrivialIdentity} would make no mathematical sense and the present construction would not be consistent. 
Note that this result is a natural consequence of the fact that the dynamics leaves the constraint 
 $ \bv- \bnu \times \bsigma(R) =0 $ invariant, so any \emph{smooth} function of $g(\bv- \bnu \times \bsigma(R) ) $ is  invariant by the flow. However, one needs to be more careful when $g$ is not a regular function, but a distribution. 
\end{remark}

\begin{remark}[Entropy]\label{rem:entropy}
Equation \eqref{EPVlasov2} produces the dynamics for an entropy 
\[
S = \int f\log f\,\de^3x\de^3\nu\de^3\R\de^3 v
\]
where $\de^3\R$ is the Haar measure on $SO(3)$. However, given the singular form of $f$ in (\ref{f0sol}), this expression poses severe problems in the definition of Gibbs entropy $S$ and in its dynamics. Indeed, one can use \eqref{EP-Vlasov} and (\ref{f0sol}) to write
\[
\frac{\de S}{\de t}=-\int\! f\operatorname{div}X\,\de^3x\,\de^3\nu\,\de^3\R\,\de^3 v
=-\int \!\phi\left(\left.\operatorname{div}X\right)\right|_{\bv=\bnu \times \bsigma(\R)}\de^3x\,\de^3\nu\,\de^3\R\,.
\]
Unfortunately, the quantity $\left(\left.\operatorname{div}X\right)\right|_{\bv=\bnu \times \bsigma(\R)}$ is undetermined; since all that is known from (\ref{f0sol}) is the value $\left.X\right|_{\bv=\bnu \times \bsigma(\R)}$. 
At this point, one may be tempted to require  $\left(\left.\operatorname{div}X\right)\right|_{\bv=\bnu \times \bsigma(\R)}=0$ for physical reasons. However, the exact mathematical (and even physical) interpretation of this condition is not clear to us, thus we we shall not discuss this possibility further here. 
 Notice that this indeterminacy problem is not related to the compressibility of the phase space flow and thus it cannot be overcome by the insertion of a metric tensor, as in equation (\ref{kineqmetric}).
\end{remark}
\color{black}

A more convenient form of writing equation  (\ref{EPVlasov-final}) may be obtained by introducing the {\color{black}density} variable
\[
\varphi(\bx,\bnu,j,\bn,t):=\phi(\bx,\bnu,\R,t)
\]
where $j=\R J\R^T$ and $\bn=\R\bchi$. Taking the differential of the above definition yields
\[
\pp{\phi}{\R}\R^T=\left(\pp{\varphi}{\bn}\,\bn^T\!\right)^\mathcal{\!A\!}+\left[\pp{\varphi}{j},j\right]
\]
so that the kinetic equation reads as
\begin{framed}
\begin{equation}
\frac{\partial \varphi}{\partial t}
+
\bnu\times\bsigma\cdot\pp{\varphi}{\bx}
+
\boldsymbol\nu\times\bn\cdot\pp{\varphi}{\bn}
+
\operatorname{Tr}\!\left(\!
\left[\widehat{\nu},j\right]\pp{\varphi}{j}\right)
+
\pp{}{\boldsymbol\nu}\cdot\left(\varphi\,\mathbf{a}_{\boldsymbol\nu}\right)
=0 \, . 
\label{EPVlasov-final2}
\end{equation}
\end{framed}\noindent
The above equation emerges from the moment dynamics for equation \eqref{EPVlasov2} and incorporates its physical content.  Strictly speaking, the moment equation (\ref{EPVlasov-final2}) cannot be interpreted as a kinetic equation on its own, because the space on which it is defined is not  an appropriate phase space in position-velocity coordinates (there is no velocity associated to the spatial coordinate $\bf x$). 
In particular, a relevant difficulty emerges when one tries to redefine Gibbs entropy in terms of the variable $\varphi$ alone. Namely, conservation of $\int \!\varphi\log\varphi\,\de^3x\,\de^3\nu\,\de^6 j\,\de^3n$ does not hold and one is forced to adopt \emph{ad hoc} methods such as the metric tensor approach  \cite{TuMuMa1999} outlined in the Introduction; see equation \eqref{kineqmetric}.  However, the explicit implementation of the metric tensor approach on a space that is not a phase space  is beyond the scope of this paper. We just note here that in contrast to these earlier works we failed to find an explicit expression for the metric tensor for our case. While it may be possible that such an expression could be found, we have left this interesting exploration to future studies. 

%

Here, we  assume $U(\bx-\bx',\R'\R^T)=\mathcal{U}(\bx-\bx',\bn\cdot\bn')$, so that
\begin{multline}
a_{\bnu}(\bx,\bv,\bnu,\bn,j)=\left(j+\widehat\sigma\widehat\sigma\right)^{-1}\left(j\bnu\times\bnu
+\bE\times\bn-\bn\times\partial_\bn \mathcal{U}* \!\int\!\varphi\,\de\bnu\de j \right.  
\\
\left.
+\bsigma\times(\bnu\times\bnu\times\bn)+\bsigma\times\partial_\bx \mathcal{U}*{\color{black}\!\int\!\varphi\,\de\bnu\de j}
\right) 
\label{anueq}
\end{multline}
 where we write, for brevity in notation,
\[
\mathcal{U}*\!\int\!\varphi\,\de\bnu\de j=\int\!\mathcal{U}(\bx-\bx',\bn\cdot\bn')\,\varphi(\bx,\bnu,\bn,j)\,\de\bnu\de j\,\de\bx'\de\bn'
\,.\]

\begin{remark}[Individual particle solutions]\label{rem:individual} 
One may verify that equation (\ref{EPVlasov2}) admits the (Klimontovich) single-particle solutions of the form 
\begin{equation} 
\varphi=\delta(\bx-\boldsymbol{X}(t) ) \delta(\bnu - \boldsymbol{\mathcal{V}}(t) ) \delta (\bn - \mathbf{N}(t) )  \delta(j-\mathcal{J}(t) ) 
\label{Klimontovich}
\end{equation} 
where $\boldsymbol{X}(t)$, $\boldsymbol{\mathcal{V}}(t)$, $\mathbf{N}(t)$, $\mathcal{J}(t)$ satisfy the single particle solutions for 
the individual ball (\ref{nueq}, \ref{neq}, \ref{chaplygin-ball-eqs}), and the rolling constraint 
\[ 
\dot{\boldsymbol{X}}(t) =  \boldsymbol{\mathcal{V}} \times \boldsymbol{\sigma}(\mathbf{N}) \, .
\] 
\end{remark}

\subsection{Conservation of the Jellett quantity} 

Let us now move away from considering an individual particle, and move towards the assembly of particle given by the Vlasov equation. Since the motion of individual particle has particular conservation laws (under proper assumptions corresponding to Chaplygin integrals), it is interesting to see the modification of these conservation laws in the kinetic framework. 

In what follows, we shall concentrate on the 
microscopic Jellett quantity 
\begin{equation}
q_j= \big(j \bnu\big) \cdot  \bsigma(\bn) \, . 
\label{microJellet} 
\end{equation} 
As discussed in Section \ref{indiv-ball-sec}, the Jellet conservation law plays an important role in our considerations of a kinetic theory for multi-particle dynamics. 
For a microscopic particle,  we  write the following relation
\begin{equation} 
\frac{d q_j}{d t}= \frac{d}{d t} \big<j \bnu, {\bsigma}(\bn) \big>=
\big<   \dot j \bnu + j \dot \bnu \, , \, \bsigma(\bn) \big> 
  + \big< j \bnu, \pp{\bsigma}{\bn} \dot \bn \big>  =0 
\label{jell0} 
\end{equation} 
due to the conservation of the Jellett integral for an individual ball. 
In terms of the kinetic approach, equation (\ref{jell0}) is reformulated as the identity 
\begin{equation} 
\big(   - [j,\hat{\boldsymbol{\nu}}] \bnu + j \mathbf{a}_{\bnu} \big) \cdot  \bsigma(\bn)
  -  j \bnu \cdot  \big(\pp{\bsigma}{\bn}   \cdot \bnu \times \bn\big) =0 
\,,\label{jellkin} 
\end{equation} 
in which $a_{\bnu}$ replaces $\boldsymbol{\dot\nu}$ which is given by (\ref{anueq}) and we have also used equations (\ref{jeq},\ref{neq}) to substitute for $\dot j$ and $\bn$, respectively. 
It is more illuminating to write this expression as 
\begin{equation} 
 \frac{\partial q_j}{\partial \bn} \cdot {\bnu \times \bn}+ 
 \frac{\partial q_j}{\partial \bnu} \cdot {\mathbf{a}_{\bnu}}
+
{\rm Tr} \Big( [j, \hat{\nu} ] \frac{\partial q_j}{\partial j}\Big) =0 \, , 
 \label{cons-gen}
\end{equation}
as it is true for any kinetic generalization of a conserved quantity $q_j$. This quantity is essentially a generalization of a microscopic time derivative. 
 We define the local Jellett density  $Q_j(\bx,t)$ as 
\begin{equation} 
Q_j(\bx,t)=\int q_j (j, \bnu, \bn )  \varphi(\bx, t, \bnu, \bn{,j})   \mbox{d} \bnu\, \mbox{d} \bn\, \mbox{d} j
\label{jelldef}
\end{equation} 
and compute its rate of change as follows: 
\begin{align} 
\pp{Q_j}{t} &= \int  q_j \pp{\varphi(\bx, t, \bnu, \bn{,j}) }{t}  \mbox{d} \bnu \mbox{d} \bn \mbox{d} j
\nonumber 
\\ 
&=- \int q_j \Big( \bnu\times\bsigma\cdot\pp{\varphi}{\bx}
+
\boldsymbol\nu\times\bn\cdot\pp{\varphi}{\bn}
-\operatorname{Tr}\!\left(\!
\left[j,\widehat{\nu}\right]\pp{\varphi}{j}\right)
+
\pp{}{\boldsymbol\nu}\cdot\left(\varphi\,\mathbf{a}_{\bnu}\right)
\Big) 
 \mbox{d} \bnu \mbox{d} \bn \mbox{d} j
\nonumber 
\\ 
&=-\, \pp{}{\bx}  \int   \big(q_j \varphi \bnu \times \bsigma  \big) 
\mbox{d} \bnu \mbox{d} \bn \mbox{d} j + 
\int q_j \varphi \left\{ {\rm div}_{\bn} (\bnu \times \bn) +{\rm div}_{j} [j, \hat{\nu} ] \right\} 
\nonumber 
\\
& \quad + \int \varphi \Big[ 
\pp{q_j}{\bn} \cdot {\bnu \times \bn}+ 
 \pp{q_j}{\bnu} \cdot {\mathbf{a}_{\bnu}}
+
{\rm Tr} \Big( [j, \hat{\nu} ] \pp{q_j}{j}\Big) \Big] \mbox{d} \bnu \mbox{d} \bn \mbox{d} j \, . 
\nonumber 
\end{align} 
Since the  term in square brackets on the right-hand side vanishes due to (\ref{jellkin}), and the divergencies inside the curly bracket terms vanish because of the antisymmetry conditions, we conclude that 
\begin{equation} 
\pp{Q_j}{t}=- \,\pp{}{\bx} 
\int   \big(q_j \varphi \bnu \times \bsigma  \big) 
\mbox{d} \bnu \mbox{d} \bn  \mbox{d} j
\, . 
\label{jellcons}
\end{equation} 
\begin{remark}
Note that an exact conservation law for the Jellett quantity no longer exists in the kinetic framework. This is because in kinetic theory each point in space contains many particles. While the Jellett integral is conserved for each individual particle, the most general conclusion one can reach for an assembly of particles is the continuity equation (\ref{jellcons}). 
\end{remark} 

\begin{remark}[Jellett integral for individual particle solutions]
One can verify that for the individual particle solutions given by (\ref{Klimontovich}), the Jellett integral is conserved exactly, \emph{i.e.} 
\begin{equation} 
\frac{d Q_j (\boldsymbol{X}(t),t)}{d t}=0 \quad \mbox{along} \quad \frac{d \boldsymbol{X}}{dt}:= \boldsymbol{\mathcal{V}} \times \bsigma(\mathbf{N}) \, .
\label{jellconslagr}
\end{equation} 
This is the conservation law of a 3-form density (volume) in the Lagrangian coordinates.  Here the notation follows directly from  \eqref{Klimontovich}.
\end{remark}
Let us extend this result for a general conservation law that is satisfied by an individual particle.   As we have mentioned before, if the particles do not interact, and satisfy Chaplygin's conditions formulated above, then each particle has three conservation laws: energy $q_e$, Jellett $q_j$ and Chaplygin  (or Routh) integral $q_r$. In general, if there is a conservation law for an individual particle, it generalizes to the continuum conservation law as follows. 
\begin{theorem}[Conservation laws for nonholonomic kinetics]
Suppose $q(j,\bnu, \bn)$ is a conserved quantity for the motion of individual particle, \emph{\i.e.} 
${d q}/{dt}=0$ when $\nu$ satisfies (\ref{neq},\ref{nueq2}). \footnote{It is important that the quantity $q$ does not depend on $\bx$.} 
Define 
\begin{equation}
Q(\bx,t)=\int q(j,\bnu, \bn) \varphi(t, \bx, \bnu, \bn{,j})  \mbox{d} \bnu \mbox{d} \bn \mbox{d} j \, .
\label{Qdef} 
\end{equation} 
Then, $Q(\bx,t)$ satisfies the continuity equation
\begin{equation} 
\pp{Q}{t}=-\pp{}{\bx} \int  \big(q \varphi  \bnu \times \bsigma(\bn) ) \mbox{d} \bnu \mbox{d} \bn \mbox{d} j \, . 
\label{Qcons}
\end{equation} 
Moreover, $Q(\bx,t)$ is conserved exactly on the individual particle solutions (\ref{Klimontovich}). 
\end{theorem}
\noindent 
{\bf Proof} The proof follows the proof for Jellett conservation law.

In order to compute the corresponding conservation laws for the energy and Chaplygin ball, one defines  the corresponding energy density 
\[ 
Q_e(\bx,t)=\int q_e(\bnu,\bn,j)   \varphi(t, \bx, \bnu, \bn, j)  \mbox{d} \bnu \mbox{d} \bn  \mbox{d} j
\] 
and Chaplygin density 
\[ 
Q_c(\bx,t)=\int q_c(\bnu,\bn,j)   \varphi(t, \bx, \bnu, \bn, j)  \mbox{d} \bnu \mbox{d} \bn \mbox{d} j \, . 
\]
This leads to the following Corollary. 
\begin{corollary} 
The energy density $Q_e$ and Chaplygin integral $Q_c$ satisfy the following conservation laws 
\begin{equation} 
\pp{Q_e}{t}=-\pp{}{\bx} \int  \big(q_e \varphi \bnu \times \bsigma(\bn)) \mbox{d} \bnu \mbox{d} \bn \mbox{d} j \, , 
\label{Econs}
\end{equation} 
\begin{equation} 
\pp{Q_c}{t}=-\pp{}{\bx} \int  \big(q_c \varphi \bnu \times \bsigma(\bn)) \mbox{d} \bnu \mbox{d} \bn \mbox{d} j\, . 
\label{Qcons2}
\end{equation} 
\end{corollary} 
Note that for the particles possessing arbitrary radial interactions through their centers of mass, the Chaplygin integral and energy are \emph{not} conserved. However, the Jellett integral is still conserved for every particle. Consequently, we have paid particular attention to the Jellett integral. 

\begin{remark} 
Note that the momentum is not conserved for the system we consider. Neither an individual ball, nor the system as a whole conserves momentum, due to the rolling constraint. We shall thus avoid deriving equations for the momentum for the fluid of rolling particle, since they do not have the physical justification invoked in the motion of ``regular" fluid.  
\end{remark}

\section{Cold fluid closure}
\label{sec:cold-fluid}
\subsection{Cold fluid equations of motion} 
Although the conservation laws derived   in the previous section may be useful, they are not sufficient to close the system and reduce the motion to observables, \emph{i.e.}, to write a fluid-like equations in $\bx$ and $t$. In order to obtain such a closure, we use the traditional moment method for cold plasma. 
 In particular, we compute the equations for the moments
\[
\int\varphi\,\de n\,\de j\,\de\nu\,,\qquad\int\!\bnu\,\varphi\,\de n\,\de j\,\de\nu\,,\qquad\int\!\bn\,\varphi\,\de n\,\de j\,\de\nu\,,\qquad\int\!j\,\varphi\,\de n\,\de j\,\de\nu
\,.\]
These equations are found to be
\begin{align*}
&\frac{\partial}{\partial t}\int\!\varphi\,\de n\,\de j\,\de\nu+\frac{\partial}{\partial \bx}\cdot\!\int\!(\bnu\times\bsigma(\bn))\,\varphi\,\de n\,\de j\,\de\nu=0
\,,\\&
\frac{\partial}{\partial t}\int\!\bnu\,\varphi\,\de n\,\de j\,\de\nu+\frac{\partial}{\partial \bx}\cdot\!\int\!(\bnu\times\bsigma(\bn))\bnu\,\varphi\,\de n\,\de j\,\de\nu-\int\!a_{\bnu}\,\varphi\,\de n\,\de j\,\de\nu=0
\,,\\&
\frac{\partial}{\partial t}\int\!\bn\,\varphi\,\de n\,\de j\,\de\nu+\frac{\partial}{\partial \bx}\cdot\!\int\!(\bnu\times\bsigma(\bn))\bn\,\varphi\,\de n\,\de j\,\de\nu-\int\!\varphi\,\bnu\times\bn\,\de n\,\de j\,\de\nu=0
\,,\\&
\frac{\partial}{\partial t}\int\!j\,\varphi\,\de n\,\de j\,\de\nu+\frac{\partial}{\partial \bx}\cdot\!\int\!(\bnu\times\bsigma(\bn))j\,\varphi\,\de n\,\de j\,\de\nu-\int\!\varphi\left[\widehat\bnu,j\right]\de n\,\de j\,\de\nu=0
\,.
\end{align*}
The last terms in the third and fourth moment equations arise from integration by parts. 
At this point, one may close the system by invoking the {\bfi cold-fluid ansatz},
\begin{equation}
\varphi(\bx,\bnu,\bn,j,t)=\rho(\bx,t)\,\delta(\bnu-\bomega(\bx,t))\,\delta(\bn-\boldsymbol{n}(\bx,t))\,\delta(j-\mathcal{J}(\bx,t)) \, . 
\label{coldfluid-Ansatz}
\end{equation}
The cold fluid equations follow from this ansatz, dividing the last three moment equations by $\rho$. The dynamics of the density $\rho$ is then computed directly, leading to:
\begin{align}
&\frac{\partial\rho}{\partial t}+\nabla\cdot(\rho\,\bomega\times\bsigma(\boldsymbol{n}))=0
\,,\label{rhoeq}
\\&
\frac{\partial\bomega}{\partial t}+(\bomega\times\bsigma(\boldsymbol{n})\cdot\nabla)\bomega=\boldsymbol{a}
\,,\label{omegaeq} 
\\&
\frac{\partial\boldsymbol{n}}{\partial t}+(\bomega\times\bsigma(\boldsymbol{n})\cdot\nabla)\bn=\bomega\times\boldsymbol{n}
\,,\label{neq2} 
\\&
\frac{\partial\mathcal{J}}{\partial t}+(\bomega\times\bsigma(\boldsymbol{n})\cdot\nabla)\mathcal{J}=\left[\widehat\bomega,\mathcal{J}\right]
\,.\label{jeq0} 
\end{align}
The angular acceleration $\boldsymbol{a}$ on the right-hand side of the moment equation (\ref{omegaeq}) for the evolution of spatial angular velocity $\bomega(\bx,t)$ in the cold fluid ansatz (\ref{coldfluid-Ansatz}) is given explicitly by
\begin{multline}
\boldsymbol{a}(\bx,t)=\big(\mathcal{J}+\widehat\sigma(\boldsymbol{n})\widehat\sigma(\boldsymbol{n})\big)^{-1}
\bigg(\mathcal{J}\bomega\times\bomega
-\boldsymbol{n}\times\partial_{\boldsymbol{n}}\!
\int\!\rho(\bx')\,{\color{black}\mathcal{U}_2(\boldsymbol{n}(\bx)\cdot\boldsymbol{n}(\bx'))}\,\de x'
\\
+{\color{black}\mathbf{E}}\times\boldsymbol{n}+\bsigma(\boldsymbol{n})\times\left(\bomega\times\bomega\times\boldsymbol{n}+\nabla{\color{black} \mathcal{U}_1}*\rho\right)
\bigg)
\,.\label{aeq}
\end{multline}
Here the $*$ symbol denotes convolution {\color{black} and we have chosen a potential $\mathcal{U}(\bx-\bx',\bn\cdot\bn')$ of the additive form
\begin{equation}
\mathcal{U}(\bx-\bx',\bn\cdot\bn')=\mathcal{U}_1(\bx-\bx')+\mathcal{U}_2(\bn\cdot\bn')
\,.
\label{interact-pot}
\end{equation}
}
The mass density $\rho$ affects the angular acceleration in equation (\ref{aeq})  only through the interaction potentials.   Finally, the angular velocity evolution equation (\ref{omegaeq}) may be assembled as
\begin{multline}
\big(\mathcal{J}+\widehat\sigma(\boldsymbol{n})\widehat\sigma(\boldsymbol{n})\big)\left(\frac{\partial\bomega}{\partial t}+(\bomega\times\bsigma(\boldsymbol{n})\cdot\nabla)\bomega\right)=
\\
\mathcal{J}\bomega\times\bomega
-\boldsymbol{n}\times\partial_{\boldsymbol{n}}\!
\int\!\rho(\bx')\,\mathcal{U}_{2}(\boldsymbol{n}(\bx)\cdot\boldsymbol{n}(\bx'))\,\de x'
+{\mathbf{E}\times\boldsymbol{n}}+\bsigma(\boldsymbol{n})\times\left(\bomega\times\bomega\times\boldsymbol{n}+\nabla \mathcal{U}_{1}*\rho\right)
\,.\label{accel}
\end{multline}
Equations (\ref{rhoeq}, {\ref{neq2}, \ref{jeq0}) and (\ref{accel}) provide a closed system. In these equations, the nonholonomic constraint from kinetic theory enters equation (\ref{omegaeq}) through the relation for acceleration (\ref{aeq}). 

The cold plasma motion equation (\ref{accel}) is formulated in terms of observable variables. One may now seek the solutions of the cold plasma closure equations. Of particular interest would be the flows in confined geometries, such as straight channel channel. This direction of research will be explored in our further studies, as the main focus of this paper is the kinetic theory.  Here, we present an interesting solution of the \emph{full} kinetic equations, that is relevant for the cold plasma-like flows of the Poiseulle type. 

\subsection{Exact solution of {kinetic} equation (\ref{EPVlasov-final2})
in the cold fluid class}
We shall show how to write exact solutions of {\color{black}(\ref{EPVlasov-final2}) and (\ref{anueq})} in a geometrical setting similar to the Poiseulle flow. Namely, let us consider a statistical ensemble of rolling balls, whose direction of motion is along the $x_1$-axis, and whose  solution is independent of $x_1$, but depends on $x_2$. 
Let us first look at the ensemble from the microscopic point of view. We need to enforce that a microscopic particle rolls along the $x_1$-axis with a constant speed. This is possible to achieve in Chaplygin case, when each particle is symmetric, so $i_1=i_2$ (two components of the microinertia tensor are equal), and the third axis of inertia is collinear with the director from the center of mass to the geometric center, \emph{i.e.}, $\bchi \parallel \mathbf{e}_3$. If these assumptions about the properties of the microscopic particles are satisfied, the rotation about $\mathbf{e}_3$-axis leaves the tensor of inertia invariant, so $j=i$.  
This result is afforded by the following 
\begin{lemma}[Rotation about the inertia axis of symmetry]
Suppose $\R$ is a rotation about the $\mathbf{e}_3$ axis of inertia by the angle $\alpha$, 
and the body-frame microscopic tensor of inertia is $\sf{i}=\mbox{diag}(i_1,i_2,i_3)$. Then, 
\begin{equation} 
\sf{j}=\R \sf{i} \R^T={\rm diag}\big(\sf{i}_1, \sf{i}_2,\sf{i}_3\big)+\big(i_2 -i_1 \big) 
\left[ 
\begin{array}{ccc} 
\sin^2 \alpha & - \cos \alpha \sin \alpha & 0 \\ 
- \cos \alpha \sin \alpha & \sin^2 \alpha & 0 \\ 
0  & 0 & 0 
\end{array} 
\right].
\label{rotIJ}
\end{equation} 
\end{lemma} 
\begin{proof} 
Proof of this Lemma is obtained by direct computation. 
\end{proof} 
Thus, in the Chaplygin case of a symmetric unbalanced ball, $i_1=i_2$ and $\sf{i}=\sf{j}$, so the tensor of inertia is the same in the body and spatial frame for this particular motion of rotation about the $\mathbf{e}_3$ axis. 
In addition, for such motions the position of the center of mass of this particle does not change in time, so $\bsigma=$const, and $\bnu$=const, since the ball will move indefinitely with a constant speed. Moreover, since $\bnu \parallel \mathbf{e}_3$, it is easy to see that $j \bnu \parallel \bnu$, so $j \bnu \times \bnu=0$. In addition, since the rotation is about the $\bchi$ axis, we have $\bn = \R \bchi = \bchi$. In the absence of external forces and interaction potential, all microscopic particles will move independently. Such a microscopic solution sets up a reasonable ansatz for the full solution of the kinetic equation. 

Let us now turn our attention to the kinetic equation (\ref{EPVlasov-final2}). We need to enforce the microscopic motion preserving the Poiseulle flow geometry, so we assume
\begin{equation} 
\varphi(t,\bx,\bnu,\bn,j):=\varphi_0(x_2)\delta\big(\bnu - \bnu_0(x_2) \big) \delta \big(\bn - \bn_0 (x_2) \big) \delta \big(\bnu - \bnu_0 (x_2) \big) \delta(j-\sf{i})  \, . 
\label{phiP}
\end{equation} 
Here, $\sf{i}$ is a given constant matrix, having the physical meaning of the microscopic inertia matrix as described above. 
Following the microscopic picture, we assume the functions $\bsigma_0(x_2)$ and $\bnu_0(x_2)$ to be  both perpendicular to the $x_1$-axis, and $\bnu_0 \parallel \bn$, so $\bsigma_0 \times \bnu_0 \parallel x_1$-axis. Then, $\bnu_0 \times \bn_0=0$, and since the rotation is only about the third principal axis of the tensor of inertia, one finds that $j=\sf{i}$ and $\varphi$ is independent of $j$. 

The setup of the problem in microscopic setting, is illustrated in Fig~\ref{fig:planarflow}. 
\begin{figure}
\label{fig:planarflow}
\centering
 \includegraphics[scale=.3]{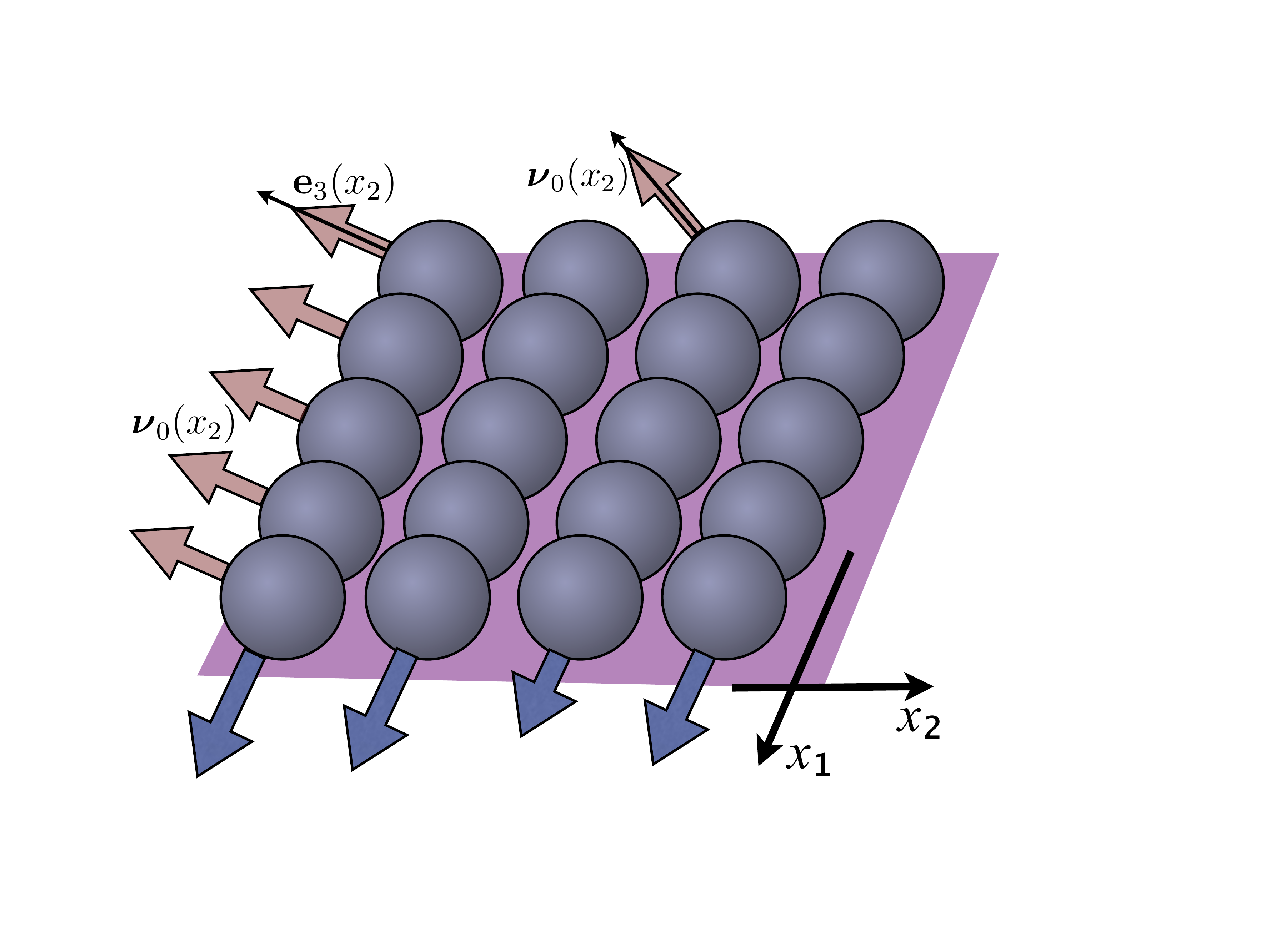}
   \caption{ Microscopic realization of Poiseulle flow. Each individual ball rolls in a straight line along $x_1$,  rotating about its axis of symmetry $\mathbf{e_3}$. The dynamical quantities are allowed to change with $x_2$, but since all the balls roll along parallel  straight lines without collisions, the solution is stationary. } 
\
\end{figure}
This leads to the simple equation
\begin{equation} 
\frac{\partial}{\partial \bnu} \varphi_0(x_2) \mathbf{a}_{\bnu}=0 \, .
\label{phieqsimple}
\end{equation} 
A possible way to satisfy (\ref{phieqsimple}) is to posit $\mathbf{a}_{\bnu}=0$. Looking at equation (\ref{anueq}), we see that $j \bnu_0 \times \bnu_0=0$, and $\mathbf{a}_{\bnu}=0$ iff 
\begin{equation} 
\bn \times \big( \bE + \partial_{\bn} \mathcal{U} * \varphi \big) =\bsigma \times \partial_{\bx} \mathcal{U}* \varphi \, . 
\label{anueq2}
\end{equation} 
 In the case $\bE=0$ (no external forcing) and $U=0$ (no interaction potential), equation (\ref{anueq}) is trivially satisfied. 

\color{black} 
Another  approach for solving equation (\ref{phieqsimple}) is to consider $\mathcal{U}=0$, but $\bE = \bE_0 \neq 0$.  The simplest condition for solutions to exist is then $\bE \parallel \bn$, \emph{i.e.}, the external field acts along the rotation axis, which is also the symmetry axis of the ball.  In this case, there is no acceleration of the microscopic particle and the equation is trivially satisfied. That case, however, is physically trivial and therefore of no interest to us.  In a more general case, $\bE$ need not be parallel to $\bn$, even though $\mathbf{a}_{\bnu} \neq 0$, since the acceleration is independent of $\bnu$. 
An even more general solution may be obtained for a central potential, \emph{i.e.}  
\[ 
\mathcal{U}\big(\bx-\bx', \bn \cdot \bn' \big) =\mathcal{U}\big(|\bx-\bx'|\big):=\mathcal{U}(\varrho) 
\,.\] 
In this case,  the convolutions in (\ref{anueq}) give $\partial_{\bn} \mathcal{U}*\varphi=0$, and 
we are left with the following integral equation for the profile: 
\begin{equation} 
\bn_0(x_2) \times \bE=\bsigma(x_2) \times \int  \frac{\bx-\bx'}{|\bx-\bx'|}\frac{\partial \mathcal{U}}{\partial \varrho} \varphi_0(x_2') \mbox{d}x_1 \mbox{d} x_2 \mbox{d} \bn' 
\,.\label{varphieq3}
\end{equation} 
One should also remember that $\bsigma(x_2)=r \z+ \bn_0(x_2)$.  It is perhaps easier to solve the equation (\ref{varphieq3}) backwards.  That is, given a function $\varphi(x_2)$, find $\bE$ such that the equation is satisfied.

Choose an arbitrary function $\varphi_0(x_2)$ that satisfies proper smoothness and integration properties. For a given $\mathcal{U}$,  define 
\[ 
\mathbf{q}(\bx)=\int \frac{\bx-\bx'}{|\bx-\bx'|}\frac{\partial \mathcal{U}}{\partial \varrho} \varphi_0(x_2') \mbox{d}x_1' \mbox{d} x_2
 \mbox{d} \bn' 
\,.\] 
If we define 
$\mathbf{b}=\bE-\mathbf{q}$ and $\mathbf{c}=R \z\times \bq$, then the equation (\ref{varphieq3}) may be rewritten as 
\begin{equation} 
\bn_0(x_2) \times \mathbf{b}(x_2)=\mathbf{c}(x_2) \, , 
\label{bceq}
\end{equation} 
with  $\bn(x_2)$ being unknown. 
The solution to this equation exists if 
\[
\mathbf{b} \cdot \mathbf{c}= \big( \bE-\mathbf{q} \big) \cdot (\z \times \mathbf{q})=0 \, , 
\]  
or simply that 
\begin{equation} 
\bE \cdot (\z \times \mathbf{q})=0 \, , 
\label{Ezcond}
\end{equation}  
 which can be satisfied if $\bE$ is chosen to be perpendicular to the plane of rolling, or $\bE \parallel \z$.   Incidentally, $\bE \parallel \z$ is the most interesting physical case, describing the case of strong gravitational attraction of rolling particles to the substrate on which they roll. 

Given that (\ref{Ezcond}) is satisfied, in order to find the solution for $\bn(x_2)$ 
we take the cross-product of $\mathbf{b}(x_2)$ with equation (\ref{bceq}), in order to obtain 
\begin{equation} 
\left( | \mathbf{b} |^2 \mbox{Id} - \mathbf{b} \mathbf{b}^T\right) \cdot \bn_0(x_2)=\mathbf{b} \times \mathbf{c} \, . 
\label{n0eq}
\end{equation} 
If, for a given $\mathbf{b}(x_2)$, a  homogeneous solution of (\ref{n0eq}) is given by $\bn_h(x_2)$, and the inhomogeneous solution to (\ref{n0eq}) is $\bn_i(x_2)$, then the general solution to (\ref{n0eq}) is simply 
\[ 
\bn_0(x_2)=C(x_2) \bn_h(x_2) +\bn_i(x_2) \, , 
\] 
where $C(x_2)$ is an arbitrary \emph{scalar} function of $x_2$.

The range of validity for the solution found in this section is narrower than that of the cold fluid closure. However, it affords substantial flexibility in the choice of velocity profile. This solution also illustrates that familiar concepts from the \emph{inviscid} two-dimensional channel flow, such as the choice of an arbitrary velocity profile, seem to carry over to the nonholonomically constrained fluid. In future work, we will address the stability of the derived solutions, which has not been addressed here.  

\begin{remark}[Connection to the cold fluid closure] 
It is interesting to connect  the solution derived here to the cold fluid closure obtained earlier in (\ref{rhoeq}-\ref{neq2}) and (\ref{accel}). In the assumptions used here, we see that both $\bomega$ and $\bsigma$ lie in the plane perpendicular to the $x_1$-axis, and all variables depend on the coordinate $x_2$
only. Thus, $\bomega \times \bsigma \cdot \nabla=0$ for all variables. On the other hand, because of the choice of the axis of rotation and the symmetry of the microscopic moment of inertia, 
$\left[\widehat\bomega,\mathcal{J}\right]=0$. Thus, the equations (\ref{rhoeq}-\ref{neq2}) are satisfied identically provided that $\mathbf{a}=0$. Equation (\ref{varphieq3}) is precisely that condition with the additional simplification of $\mathcal{U}_2=0$, so the potential interaction depends only on the Euclidian distance between the microscopic particles. 
\end{remark}

\section{Conclusions} 
This paper derives the kinetic theory of an ensemble of interacting particles that are each subject to the rolling constraint. The main difference of the present work with the previous literature in the area is that we consider the constraints applied \emph{to every particle in the ensemble}, rather than on system in general. We have shown how to derive the equations of motion using the geometric Euler-Poincar\'e  principle, with constraints treated by the Lagrange-d'Alembert principle. 
The limitation to constraints that are linear in the velocities, imposed by the Lagrange-d'Alembert principle, may be overcome by using the more general Gauss' principle of least constraint. 

We have not pursued this line of investigation further here, because the Gauss method  becomes rather cumbersome when the constraint is applied to each particle. In fact, we are not aware for any physically meaningful, nonlinear non-holonomic constraints that have been applied to every particle, and not to the system as a whole. Nevertheless, it is an interesting topic from the mathematical point of view that will be addressed elsewhere. 

The present paper showed that a probability density that is initially concentrated on a constraint distribution that is linear in velocities will remain concentrated on that distribution for all times. We are hopeful that this may also be true for constraints that are nonlinear in velocities. Recent work on the kinetics of the Vicsek model \cite{BoCaCa2011,Bo-etal-2011,CaCaRo2011} suggests that dynamical preservation of nonlinear constraint distributions may also be possible. 

\section*{Acknowledgements} 
We are grateful to Profs. G. Pavliotis and B. Leimkuhler for many fruitful and inspiring discussions on the topic.  
VP acknowledges  support from NSF Grant No. DMS-0908755 and from the Defense Threat Reduction Agency Joint Science and Technology Office for Chemical and Biological Defense (Grant No. HDTRA1-10-1-0070). DDH is grateful for partial support by the Royal Society of London Wolfson Research Merit Award  and the European Research Council Advanced Grant 267382 FCCA.

\appendix

\section{Proof of theorem \ref{theorem-f}}\label{App:proof1}
We only need to prove the second statement. 
Since $u_\bx$ and $u_\R$ do not depend on $\bx$, we have
\begin{align*}
\left(\pounds_{X\,}\dede{l}{X}\right)_{\!\!\bx}
=&
\left(X\cdot\nabla\right)\dede{l}{u_\bx}
+
\partial_\bx u_\bx^T\cdot\dede{l}{u_\bx}
+
\mathrm{Tr}\left(\!(\partial_\bx  u_\R^T)\,\dede{l}{u_\R}\right)
+
(\nabla\cdot X)\dede{l}{u_\bx}
\\
=&
\left(X\cdot\nabla\right)\dede{l}{u_\bx}+
(\nabla\cdot X)\dede{l}{u_\bx}
\end{align*}
Then, upon calculating
\[
\dede{l}{u_\bx}=fu_\bx=f\nu\bsigma(\R)
\,,\qquad
\frac{\partial}{\partial\bx}\dede{l}{f}=-\frac{\partial U}{\partial\bx}* \!\int \!f\,\de\bv\de\bnu 
\,,
\]
the nonholonomic part of the first equation becomes  (in vector notation) 
\begin{align*}
\left(\frac{\partial}{\partial t}\dede{l}{X}+\pounds_{X\,}\dede{l}{X}\right)_{\!\!\bx\!\!}\times&\ \bsigma
-
f\frac{\partial}{\partial\bx}\frac{\delta l}{\delta f}\times\bsigma=
\\
=&\
f\left((X\cdot\nabla)\left(\nu\bsigma(\R)\right)\right)\times\bsigma(\R)+f\left(\partial_\bx U* \!\int \!f\,\de\bv\de\bnu \right)\times\bsigma
\\
=&\
f\left(\left(\nu(u_\R^{ij\,}\partial_{\R^{ij}})\bsigma+a_{\bnu}\times\bsigma\right)\right)\times\bsigma+f\left(\partial_\bx U* \!\int \!f\,\de\bv\de\bnu \right)\times\bsigma
\\
=&\
f\bsigma\times\bsigma\times a_{\bnu}
-f\bsigma\times(\bnu\times\bnu\times\bn(\R))+f\left(\partial_\bx U* \!\int \!f\,\de\bv\de\bnu \right)\times\bsigma
\end{align*}
where we have used the continuity equation \eqref{EP-Vlasov}. 
Therefore, the first Euler-Poincar\'e equation reads as 
\begin{equation*}
\left(
\!\left(\frac{\partial}{\partial t}\dede{l}{X}+\pounds_{X\,}\dede{l}{X}\right)_{\!\R\!\!}\R^T-f\frac{\partial}{\partial\R}\frac{\delta l}{\delta f}\,\R^T
\right)^\mathcal{\!\!A\ } 
=
-f\left(\left(
\nu\nu\bn(\R)
+
 a_{\nu}\bsigma
+\left(\partial_\bx U* \!\int \!f\,\de\bv\de\bnu \right)\right)\bsigma^T
\right)^\mathcal{\!A\ } 
\end{equation*}
Finally, one computes the micropolar terms (right hand side above) as follows
\begin{align*}
&\left(
\!\left(\frac{\partial}{\partial t}\dede{l}{X}+\pounds_{X\,}\dede{l}{X}\right)_{\!\R\!\!}\R^T-f\frac{\partial}{\partial\R}\frac{\delta l}{\delta f}\,\R^T
\right)^\mathcal{\!A\,}=
\\
=&
\!\left(\frac{\partial}{\partial t}\dede{l}{u_\R}+
(X\cdot\nabla)\frac{\delta l}{\delta u_\R}\,\R^T+\nabla_\R {\color{black}u_x^i\frac{\delta l}{\delta u_x^i}}\,\R^T +\left(\nabla_\R u_\R^{\color{black}hk\,}\frac{\delta l}{\delta u_\R^{\color{black}hk}}\right)\!\R^T+\frac{\delta l}{\delta u_\R}\operatorname{div}X
-
\frac{\partial}{\partial\R}\frac{\delta l}{\delta f}\,\R^T
\right)^\mathcal{\!A\,}
\\
=&\,
f\!\left(
j(\R)a_\nu+\left(\nabla_\R (\nu\R)^{\color{black}hk\!\!}\left(j(\R)\nu\R\right)_{\color{black}hk}\right)\!\R^T
-\frac12
\frac{\partial}{\partial\R}\operatorname{Tr}(\nu^Tj(\R)\nu)\,\R^T-\bE\bn^{T\!}+\left(\partial_\R U* \!\int \!f\,\de\bv\de\bnu \right)\R^T
\right)^\mathcal{\!A\,}
\\
=&\,
f\!\left(
j(\R)a_\nu
-j(\R)\nu\nu
-\bE\bn^{T\!}(\R)+\left(\partial_\R U* \!\int \!f\,\de\bv\de\bnu \right)\R^T
\right)^\mathcal{\!A\,}
\end{align*}
so that, in vector notation, the first Euler-Poincar\'e equation becomes \eqref{EP-EQ-1}. $\blacksquare$}

\bibliography{nonholonomicrefs}{}
\bibliographystyle{unsrt}

\end{document}